\documentclass[11pt]{article}

\usepackage{amsmath,amsthm,amsfonts,amscd,eucal,latexsym,amssymb,mathrsfs,dsfont,cancel,mathtools,physics,tensor,comment,appendix,enumerate,enumitem,cases,wasysym}
\usepackage{pdfpages}
\pdfoutput=1 
\usepackage[normalem]{ulem}
\usepackage{pst-node}
\usepackage{epsfig,caption,float,tikz,subcaption}  
\usepackage{simplewick,bm}
\definecolor{hypercolor}{RGB}{9,80,162}
\usepackage{hyperref,cleveref}
\hypersetup{
	colorlinks = true,
	linkbordercolor = {white},
	linkcolor={hypercolor},
	citecolor={hypercolor},
	urlcolor={hypercolor}
}
\usepackage{csquotes}

\def\cA{{\mathcal A}}

\def\cF{{\mathcal F}}

\def\cM{{\mathcal M}}

\def\cW{{\mathcal W}}
\def\cK{{\mathcal K}}

\def\beq{\begin{eqnarray}}
\def\eeq{\end{eqnarray}}
\def\pa{\partial}
               		   
\def\at{\left(}               		   
\def\aq{\left[}               		   
\def\ag{\left\{}

\def\ct{\right)}              		   
              		
\def\cg{\right\}}             		   
        
\def\d{\mathrm{d}}
\def\e{\mathrm{e}}

\newtheorem{theorem}{Theorem}[section]
\newtheorem{corollary}
[theorem]{Corollary}
\newtheorem{proposition}[theorem]{Proposition}
\newtheorem{lemma}[theorem]{Lemma}
\theoremstyle{definition}

\newtheorem{remark}[theorem]{Remark}

\DeclareCaptionFont{fig}{\fontsize{10}{12}\mdseries}
\renewcommand{\vec}[1]{\mathbf{#1}}

\newcommand{\eq}{eq.\@ }
\newcommand{\eqs}{eqs.\@ }
\newcommand{\Eq}{Eq.\@ }
\newcommand{\expvalom}[1]{\expval{#1}_\omega}

\newcommand{\expvalv}[1]{\expval{#1}_0}

\usepackage[backend=bibtex,style=alphabetic,citetracker=true,sorting=nyt,maxcitenames=10,mincitenames=3,maxbibnames=15,maxalphanames=10]{biblatex}

\begin{document}

\par 
	\bigskip 
	\LARGE 
	\noindent 
	{\bf Linear stability of semiclassical theories of gravity} 
	\bigskip 
	\par 
	\rm 
	\normalsize 
	
\large
\noindent
{\bf Paolo Meda$^{1,3,a}$}, {\bf Nicola Pinamonti$^{2,3,b}$*}  \\
\par

\small
\noindent$^1$ \textit{Dipartimento di Fisica, Universit\`a di Genova, Italy}.

\smallskip

\noindent$^2$ \textit{Dipartimento di Matematica, Universit\`a di Genova, Italy}.

\smallskip

\noindent$^3$ \textit{Istituto Nazionale di Fisica Nucleare - Sezione di Genova, Italy}.

\smallskip

\noindent * corresponding author.

\smallskip

\newcommand{\myat}{{\fontfamily{pag}\selectfont @}}

\noindent E-mail:
\href{paolo.meda@ge.infn.it}{$^a$paolo.meda\myat ge.infn.it},
\href{pinamont@dima.unige.it}{$^b$pinamont\myat dima.unige.it}

\normalsize

\par

\rm\normalsize

\par 
\bigskip

\rm\normalsize 

\bigskip

\noindent 
\small 
{\bf Abstract.} The linearization of semiclassical theories of gravity is investigated in a toy model, consisting of a quantum scalar field in interaction with a second classical scalar field which plays the role of a classical background. 
This toy model mimics also the evolution induced by semiclassical Einstein equations, such as the one which describes the early universe in the cosmological case. The equations governing the dynamics of linear perturbations around simple exact solutions of this toy model are analyzed by constructing the corresponding retarded fundamental solutions, and by discussing the corresponding initial value problem. It is shown that, if the quantum field which drives the back-reaction to the classical background is massive, then there are choices of the renormalization parameters for which the linear perturbations with compact spatial support decay polynomially in time for large times, thus indicating stability of the underlying semiclassical solution.

\bigskip

\section{Introduction}
\label{sec:intro}

In semiclassical theories of gravity, the back-reaction of quantum matter fields is studied by means of the so-called semiclassical Einstein equations, where the Einstein tensor of a four-dimensional spacetime $(\cM,g)$ is equated to the expectation value of the stress-energy tensor of matter fields evaluated on a quantum state $\omega$. Any couple formed by a spacetime and a quantum state satisfying these equations constitutes a solution in semiclassical gravity. The question about existence and uniqueness of solutions was analyzed in several physical scenarios and with different approaches, which take advantage of recent developments in the study of locally covariant quantum field theories on globally hyperbolic spacetimes \cite{Dappiaggi2008cosm,Pinamonti2011init,Pinamonti2015glob,Hack2016cosm,Juarez-Aubry2019desitt,Juarez-Aubry2021stat,Sanders2020stat,Gottschalk2021cosm,Meda2021exi,Gottschalk2022sitter,Gottschalk2022num,Juarez-Aubry2022conf} (see also \cite{Meda2021evap} for an application of the semiclassical Einstein equations in black hole evaporation). In this framework, the renormalization of quadratic fields like the Wick square or the stress-energy tensor $T_{\mu\nu}$ entering the semiclassical models is always guaranteed when sufficiently regular states are taken into account. This is the case of Hadamard states, which fulfil the microlocal spectrum condition and possess two-point functions which have an universal singular structure \cite{Radzikowski1996micro,Brunetti1996wick}.
Hence, in Hadamard states
Wick observables can be constructed by a covariant point-splitting regularisation which removes the universal singular part, and generalises the usual normal-ordering prescription. For further details about this topic we refer to \cite{Wald1977back,Kay1991uniq,Brunetti2000micro,Hollands2001wick,Hollands2002top,Hollands2005stress} (see also \cite{Hollands2015qft} for a recent review and some physical applications). 

The study of the stability of solutions of semiclassical equations remains problematic even on this class of states and even at linear order, because of the presence of higher-order derivative terms in the expectation value of the quantum stress-energy tensor. In the past years, the issue about stability of solutions of the semiclassical Einstein equations was addressed in several works \cite{Horowitz1978inst,Horowitz1980wf,Kay1981see,Yamagishi1982inst,Jordan1987stab,Suen1989unst,Matsui20inst}. 
It was argued that, because of the higher order derivative terms,
linearized semiclassical Einstein equations around chosen backgrounds
admit exponentially growing solutions. These exponentially growing linear perturbations are called runaway solutions and 
indicate that the chosen background is unstable. It is remarkable that runaway solutions are present even on flat spacetimes.
A prescription of reduction of order was presented in \cite{Simon1991stab,Parker1993red,Flanagan1996back} to eliminate runaway unstable solutions, whereas a criterion for the validity of semiclassical gravity in the linear regime was proposed in \cite{Anderson2003lr}. More recently, the stability of semiclassical solutions has been treated in the framework of the so-called stochastic gravity (see \cite{Hu2020stoc} and references therein). 
The irregularity issues of the semiclassical Einstein equations were also deeply studied in \cite{Meda2021exi} in the case of cosmological spacetimes, for arbitrary values of the coupling-to-curvature parameter. In this case, higher order derivatives of the metric appear, and furthermore the expectation value of the traced stress-energy tensor contains a non local quantum contribution at the linear order in the perturbative potential. This term represents the source of instability of the model, and it forbids to solve the semiclassical equations in a direct way. However, it was proved that unique local solutions exist after applying to the semiclassical equation an inversion formula associated to that unbounded operator. 

\bigskip

In this paper, we shall take inspiration from the ideas presented in \cite{Randjbar-Daemi980ren} and \cite{Juarez-Aubry2020init}, and we analyse the stability issue at linear order in a simple semiclassical toy model in flat spacetime, consisting of a classical background scalar field $\psi$ coupled to a quantum free scalar field $\phi$. 
The system of equations governing the dynamics is displayed in \eq \eqref{eq:system-semicl}, and the back-reaction of the quantum field $\phi$ on $\psi$ is estimated by substituting the 
classical field $\phi^2$ 
with $\expvalom{\phi^2}$, the expectation value of the quantum Wick square in the quantum state $\omega$.
With this picture in mind, we are interested in analyzing the stability of the solutions of this semiclassical system against linear perturbations. To this end, we discuss the equation obtained by linearizing the semiclassical equations over full solutions, which are formed by a quantum state $\omega$ for the quantum field $\phi$ and by a classical background $\psi_0$ for the classical field $\psi$. 
For the sake of simplicity, in this toy model, the state which is chosen is the Minkowski vacuum, however, similar results hold with other choices for $\omega$.
The obtained equation \eqref{eq:SCE-psi1} is a linear equation for the perturbations $\psi_1$ over the background classical field $\psi_0$, and in some cases it gives origin to a well posed initial value problem for spatially compact solutions 
despite the presence of certain unavoidable non local contributions in it. 
These non local contributions are manifest in the linear equation for $\psi_1$ written in \eq \eqref{eq:SCE-eq1-hom}, or in the corresponding 
non-homogeneous equation \eqref{eq:SCE-eq1-source} obtained when a smooth compactly supported source $f$ for $\psi_1$ is considered.

In a first step, we prove that this equation is of hyperbolic nature. We then explicitly construct the retarded fundamental solutions $D_R$. Out of it we write the most general retarded solution of \eq \eqref{eq:SCE-eq1-source}, we discuss the well posedness of the corresponding initial value problem, and we study the decay of these solutions for large times.

Notably, there are several choices of the parameters governing the dynamics for which every linearized solution having smooth compactly supported initial values, 
or emerging from a compactly supported source, decays polynomially in time for large times, thus showing that perturbations tend to disperse, or better to disappear in time. This result of stability is ensured by both the non-vanishing mass of the quantum field, by the spatial support of the initial data and the source.
On the contrary, in agreement with previous observations \cite{Horowitz1978inst,Horowitz1980wf,Jordan1987stab}, if the quantum field $\phi$ is massless, then solutions of the linearized semiclassical equation which grow exponentially in time may always exist, even if the initial values and the source are of compact support.

The toy model studied in this paper formally mimics both the cosmological semiclassical model studied in \cite{Meda2021exi} and the weak field theory discussed in \cite{Horowitz1980wf}, after interpreting the background solution as a degree of freedom of the metric, and the coupling constants as the renormalization freedoms appearing in a semiclassical theory of gravity. This analogy indicates that such semiclassical theories may have stable solutions at least when the quantum fields describing matter are massive. 

\bigskip
This paper is organized as follows, in the \hyperref[sec:model]{Section 2} we introduce the interacting toy model considered in this paper, and we discuss how to obtain a linearized semiclassical equation governing the back-reaction of the quantum field on the background. In the first part of the \hyperref[sec:stability]{Section 3}, we analyze the linearized equation, we construct its retarded fundamental solution and we show how to formulate a well-posed initial value problem out of it. In the last part of this section, we present a correspondence between the presented toy model and the semiclassical cosmological problem. Finally, the main results of the paper are summarized in  \hyperref[sec:conc]{Section 4}. \hyperref[sec:appendix]{Appendix A} contains some technical results about the decay at large time of certain functions. 

\section*{Notations}

In this paper, the units convention is $G = c = 1$, and the Lorentzian signature of the spacetime $(\cM,g)$ is $(-,+,+,+)$. Thus, the d'Alembert operator reads $\square_x = g^{\mu\nu} \nabla_\mu \nabla_\nu = -\pa^2_t + \Delta_\vec{x}$, where $\Delta_\vec{x}$ denotes the spatial Laplace-Beltrami operator.
We shall employ the following conventions on the Fourier transforms:
\begin{gather*}
	\hat{f}(p) = \cF\{f\}(p) = \int_{\mathbb{R}^4} f(x) \e^{i p \cdot x} \d^4 x, \qquad f(x) = \cF^{-1}\{\hat{f}\} = \frac{1}{(2\pi)^4} \int_{\mathbb{R}^4} \hat{f}(p) \e^{-i p \cdot x}\d^4 p, \\
	\tilde{f}(t,\vec{p}) = \cF_{\vec{x}}\{f\}(t,\vec{p}) = \int_{\mathbb{R}^3} f(t,\vec{x}) \e^{i \vec{p} \cdot \vec{x}} \d^3 \vec{x}, 
\end{gather*}
respectively. 
Hence, the convolution theorems read as $\cF \{ f * g \} = \cF \{f\} \cF\{ g\}$
where $*$ denotes the convolution operator such that $f * g = \int f(x-y) g(y) \d y$.

\section{The semiclassical model}
\label{sec:model}

In this work, we study the coupling between a quantum scalar field $\phi$ and another classical background scalar field $\psi$ in the Minkowski spacetime $(\cM,\eta)$.
The equations of motion of the corresponding free theory are  
\begin{subnumcases}{}
	\square \phi  - m^2 \phi =  \lambda \psi \phi, \label{eq:system-eq1} \\
	g_2 \square \psi   -  g_1 \psi   =  \lambda_1 \phi^2  -\lambda_2 \square \phi^2, \label{eq:system-eq2}
\end{subnumcases}
where $g_1,g_2,\lambda,\lambda_1,\lambda_2$ denote the real coupling constants of the theory. With the choice of $\lambda_2=0$, the system is Lagrangian. On the other hand, if $\lambda_2 \neq 0$, then there is sufficient freedom in the definition of the coupling constants to fix $\lambda_2 = 1$.
However, we do rather not impose any further constraint on the coupling constants, in order to describe as many semiclassical theories as possible. In particular, keeping $\lambda_2 \neq 0$, \eq \eqref{eq:system-eq2} mimics the form of equations arising in the study of semiclassical Einstein equations, as we shall discuss in \hyperref[sec:cosmo]{subsection 3.4}. The first equation \eqref{eq:system-eq1} is the equation of motion of the linear, Klein-Gordon like massive field $\phi$, where $m$ is the mass, and $\lambda \psi$ plays the role of an external potential. The quantization of $\phi$ follows straightforwardly once the external potential field $\psi$ is known. To this avail, we shall consider the unital $*-$algebra of field observables $\mathcal{A}$ generated by the identity $\mathbb{I}$ and the abstract smeared field $\phi(f)$, with $f \in C^{\infty}_0(\cM)$ any compactly supported smooth function \cite{Brunetti2015aqft,Haag2012local}.
The product in this $*-$algebra encodes the canonical commutation relations
\[
	[\phi(f_1),\phi(f_2)] = i \hbar \Delta(f_1,f_2) \mathbb{I}, \qquad f_1, f_2 \in C^{\infty}_0(\cM),
\]
given in terms of the causal propagator $\Delta \doteq \Delta_R - \Delta_A$, i.e., the difference of the retarded and advanced propagators uniquely obtained as fundamental solutions of \eq \eqref{eq:system-eq1}. The algebra of field observables can be enlarged to contain also Wick powers like $\phi^2$  when a properly defined normal ordering procedure is taken into account \cite{Hollands2001wick,Hollands2005stress,Moretti2003stress}. 

In light of the quantum nature of $\phi$, \eq \eqref{eq:system-eq2} can be interpreted in the semiclassical approximation, namely by taking the expectation values of $\phi^2$ in a suitable quantum state. More precisely, once a state $\omega$ on $\mathcal{A}$ is chosen for the quantum field theory, we have that
\[
	g_2 \square \psi - g_1 \psi  =  \lambda_1\expvalom{\phi^2} - \lambda_2 \square \expvalom{\phi^2},
\]
where $\expvalom{\phi^2}  \doteq \omega(\phi^2)$ is the expectation value of the properly normal ordered quantum field $\phi^2$ in the quantum state $\omega$. Thus, the semiclassical system corresponding to equations \eqref{eq:system-eq1} \eqref{eq:system-eq2} turns out to be
\begin{equation}
\label{eq:system-semicl}
	\begin{cases}
		\square \phi  - m^2 \phi =  \lambda \psi \phi, \\
		g_2 \square \psi   -  g_1 \psi   =  \lambda_1 \langle \phi^2\rangle_\omega  -\lambda_2 \square \langle \phi^2\rangle_\omega.
	\end{cases}
\end{equation}
To analyze the linearization of this system around a given simple solution of \eq \eqref{eq:system-semicl} determined by $(\psi_0,\omega)$, we decompose the classical field $\psi$ in two parts, the background contribution $\psi_0$ plus a perturbation $\psi_1$, so that 
\[
	\psi= \psi_0 + \psi_1.
\]
 The state $\omega$, or more specifically its two-point function $\hbar \Delta_{+,\omega}$ appearing in the second equation of the system \eqref{eq:system-semicl} through $\langle \phi^2\rangle_\omega$, can also be decomposed in two parts
\[
\Delta_{+,\omega} = \Delta_{+,\omega,0}+ \Delta_{+,\omega,1},
\]
i.e., the background contribution plus its perturbation. Notice that the linear order contribution in $\Delta_{+,\omega,1}$ is formed by two terms:
one which takes into account the modified evolution equation induced by $\psi_1$ and satisfied by $\Delta_{+,\omega}$, and 
one which consists in $w_1$, a symmetric bi-solution of the zeroth order equation of motion satisfied by $\phi$. 
Since $\Delta_{+,\omega,0}$ is a bi-solution of the same equation, the effect due to $w_1$ can be reabsorbed in a redefinition of the background theory. 
For this reason, in the following, we shall take into account explicitly the effect due to the modified evolution only. To control the linear contribution in $\Delta_{+,\omega,1}$ due to the change of dynamics, we shall use perturbation theory for quantum fields.
As we shall see later, if instead one takes into account the effect of $w_1$ explicitly, an inhomogenous (source) term has to be added to the linearized equation satisfied by $\psi_1$. In fact, this extra term will not alter the discussion we are going to present.

The quantization of $\phi$ is performed on the fixed background $\psi_0$ by considering the $*-$algebra of field observables $\mathcal{A}$ with a product that implements the canonical commutation relations emerging from to the zeroth-order equation  
\begin{equation}\label{eq:linear}
\square \phi  - m^2 \phi -  \lambda \psi_0 \phi = 0.
\end{equation}
Furthermore, both the quantum state $\omega$ and the background field $\psi_0$ satisfy the semiclassical equation, 
namely the second equation in the system given in \eq \eqref{eq:system-semicl}, which thus represents a constraint for the couple $(\psi_0,\omega)$. We recall here that the expectation value of $\phi^2$ in the state $\omega$ on this background theory is obtained by an ordinary point-splitting regularisation procedure 
\begin{equation}
\label{eq:point-splitting}
	\expvalom{\phi^2}^{(\text{bac})}(x) = \lim_{y\to x} \hbar\left( \Delta_{+,\omega}(x,y)-H(x,y) \right)+ c m^2 +  c\lambda\psi_0(x),
\end{equation}
where $\hbar \Delta_{+,\omega}$ is the two-point function of the state $\omega$, $H(x,y) \doteq u(x,y)/(x-y)^2 + v(x,y)\log ((x-y)^2/\mu^2)$ is the universal divergent contribution present in any Hadamard two-point function, with $\mu$ a fixed length scale. Furthermore, $c$ is a constant which encodes the regularisation freedom present in the construction of Wick powers like $\phi^2$ \cite{Hollands2001wick,Hollands2002top}. 

Notice that the influence of $\psi_1$ on $\phi$ is governed by the first equation in the system \eqref{eq:system-semicl}; for a given $\psi_1$, this equation descends from a Lagrangian for $\phi$.
In this Lagrangian, $\psi_1$ acts as a mass perturbation for $\phi$, hence we may use Lagrangian methods to analyze the influence of $\psi_1$ to $\phi$ even if the full system formed by both equations in \eq  \eqref{eq:system-semicl} is not Lagrangian.
Thanks to the principle of perturbative agreement (see e.g. \cite{Hollands2005stress, Drago2017pertagr}) this approach is equivalent to directly analyze the effect of $\psi_1$ on the two-point function, as it was done, e.g., in \cite{Pinamonti2011init,Pinamonti2015glob,Meda2021exi}, or to evaluate $\langle \phi^2 \rangle_\omega$ in \cite{Eltzner2011back} in cosmological spacetimes. Thus, we pass to analyze the influence of the perturbation $\psi_1$ on $\phi$ 
by means of perturbation theory considering the interaction Lagrangian
\begin{equation}
\label{eq:Lagrangian-int}
	\mathcal{L}_I = -\frac{\lambda}{2} \psi_1 \phi^2.
\end{equation}
A perturbative construction of interacting quantum field theories on a generic curved spacetime $(\cM,g)$ was rigorously formulated in a local and covariant way in the framework of perturbative algebraic quantum field theory - \emph{cf.} \cite{Brunetti2000micro,Duetsch2001aqft,Hollands2001wick,Hollands2002top,Fredenhagen2014kms,Fredenhagen2016paqft,Gere2016anly} to which we refer for further details on the construction we are going to present.
In particular, the expectation value of the interacting field $\phi^2$ in the state $\omega$ is obtained by means of the Bogoliubov map
\[
\langle \phi^2\rangle_\omega = \omega(R_V(\phi^2)),
\]
where the perturbative potential 
\[
V \doteq \int_\cM \mathcal{L}_I(x) f(x) \d^4 x = -\frac{\lambda}{2} \int_\cM \phi^2(x) \psi_1(x) f(x) \d^4 x
\]
is obtained by smearing the interaction Lagrangian given in \eq \eqref{eq:Lagrangian-int} with a smooth cut-off $f \in C^\infty_0(\cM)$ which is equal to $1$ on the compact spacetime region where we want to test the semiclassical equation. This cut-off $f$ is eventually removed by considering a suitable limit in which $f$ tends to $1$ on every point of $\mathcal{M}$.

The Bogoliubov map $R_V$ is used to represent local field observables of the interacting theory as formal power series in $\lambda$, whose coefficients are well defined elements of $\mathcal{A}$, the extended algebra of free fields.
In particular, 
\begin{equation}
\label{eq:Bogoliubov}
	R_V(\phi^2) = S(V)^{-1} T(S(V)\phi^2),
\end{equation}
where $T$ is the map which realizes the time ordering, while $S(V)$ is the time ordered exponential of the smeared interaction Lagrangian $V$, namely
\[
	S(V) = T \at \exp(\frac{i}{\hbar} V) \ct = \sum_{n \geq 0} \frac{i^n}{\hbar^n n!} T \at V^{n} \ct.
\]
For the precise construction of the time ordering map $T$ we refer to \cite{Brunetti2000micro,Hollands2001wick,Hollands2002top}, where the old construction of Epstein and Glaser in \cite{Epstein1973loc} is generalized to a generic curved spacetime. 

\medskip

With the perturbative construction of $\expvalom{\phi^2}$ at disposal, the perturbative expansion of $\expvalom{\phi^2}$ in the interacting theory obtained by means of the Bogoliubov map reads at the linear order in $V$ as   
\[
	\expvalom{\phi^2} = \expvalom{\phi^2}^{\text{(bac)}} + \expvalom{\phi^2}^{\text{(lin)}} + \dots,
\]	 
where the contributions higher than the linear one are not displayed. Here,
\begin{equation}
\label{eq:phi2-bac-lin} 
	\expvalom{\phi^2}^{\text{(bac)}} \doteq \omega(\phi^2), \qquad \expvalom{\phi^2}^{\text{(lin)}} \doteq \frac{i}{\hbar} \at \omega( T\at V \phi^2 \ct) - \omega(V \phi^2) \ct.
\end{equation}	
The factor $i$ at the right hand side of the second equation in \eq \eqref{eq:phi2-bac-lin} is present because $S(V)$ in \eq \eqref{eq:Bogoliubov} is formally unitary.  As expected by the principle of perturbative agreement, one can notice by direct computation that  $\expvalom{\phi^2}^{\text{(lin)}}$
matches the linear contribution 
obtained in \cite{Pinamonti2011init,Pinamonti2015glob, Gottschalk2021cosm, Meda2021exi} in the cosmological case.

The linearization studied in this paper consists of studying 
the semiclassical theory described by the system of equations given in \eq \eqref{eq:system-semicl}, where the expectation value of the Wick square is approximated by truncating at first order the formal power series in the interaction Lagrangian \eqref{eq:Lagrangian-int} occurring in the map $R_V$ defined in \eq \eqref{eq:Bogoliubov}. 
The state for the interacting quantum theory is constructed by means of the state on the linear quantum theory, i.e., on $\cA$, and it is fixed once and forever, no matter the form of the linear perturbation $\psi_1$.

Hence, on the one side the background theory is described by $(\psi_0,\omega)$, in which $\psi_0$ fulfils the semiclassical equation
\begin{equation}
\label{eq:background}
	g_2 \square \psi_0  - g_1 \psi_0  = (\lambda_1 -\lambda_2 \square) \expvalom{\phi^2}^{\text{(bac)}},
\end{equation}
and the the quantum state $\omega$ constrained  by \eq \eqref{eq:background} is  fixed once and for all in the linear algebra $\mathcal{A}$ of fields satisfying \eq \eqref{eq:linear}. On the other side, the linear perturbation theory of the background is described by the classical field $\psi_1$, which fulfils the following linearized semiclassical equation 
\begin{equation}
\label{eq:SCE-psi1}
	g_2 \square \psi_1 - g_1 \psi_1  = \left( \lambda_1 -\lambda_2 \square\right) \expvalom{\phi^2}^{\text{(lin)}},
\end{equation}
where $\expvalom{\phi^2}^{\text{(lin)}}$ was constructed in \eq \eqref{eq:phi2-bac-lin}. \Eq \eqref{eq:SCE-psi1} is the only equation we have at linear order, and it must be seen as a dynamical equation for the linear perturbation $\psi_1$.
If a perturbation of the state which modifies $\omega$ to $\omega+\omega_1$
is considered explicitly at the linear order, then an inhomogoenous contribution appears in \eq \eqref{eq:SCE-psi1}, which consists in adding $\omega_1(\phi^2(x)) = \hbar w_1(x,x)$ to $\expvalom{\phi^2}^{\text{(lin)}}$. This extra contribution does not depend on $\psi_1$, and hence it acts as a source term in \eq \eqref{eq:SCE-psi1}. We shall see in the next section that also solutions of the inhomogeneous version of \eq \eqref{eq:SCE-psi1} have nice decay properties for several choices of the parameters $a,\lambda, \lambda_i, g_i$.

\section{Stability of solutions of the linearized semiclassical equation}
\label{sec:stability}

With the results presented in \cite{Randjbar-Daemi980ren} in mind, our goal in this section is to show that 
perturbations $\psi_1$ over $(\psi_0,\omega)$ which solves \eq \eqref{eq:system-semicl} decay for large time at linear order in perturbation theory and for proper choices of the coupling constants of the model. To achieve this, it is crucial to assume compactly supported initial data on the fields, and to consider perturbations around  the background solution which have spatial compact support, in order to avoid the class of exponentially growing solutions already seen, e.g., in \cite{Horowitz1978inst,Horowitz1980wf,Jordan1987stab}.

\subsection{Zeroth-order solution}

Before discussing the perturbative construction of $\phi$, we give here some details on the chosen background solution $(\psi_0,\omega)$. For the sake of simplicity, we shall assume that the background $\psi_0\in \mathbb{R}$ is constant, hence the state $\omega$ must be such 
that the Wick square has constant expectation value
\begin{equation}
\label{eq:const}
	\omega(\phi^2) = \expvalom{\phi^2}^{\text{(bac)}}= -\frac{g_1}{\lambda_1} \psi_0 .
\end{equation}
In view of the renormalization freedom present in the definition of the Wick power expressed by the constant $c$ in \eq \eqref{eq:point-splitting}, it is always possible to fulfil the previous equation whenever $\omega$ is a translation invariant state, for any choice of $\psi_0 \in \mathbb{R}$. A constant background external field $\psi_0$ corresponds to a mass renormalization, that is,	
\[
	m_\lambda^2= m^2 +\lambda \psi_0 = m^2 - \lambda \frac{\lambda_1}{g_1} \langle \phi^2\rangle_{\omega}^{\text{(bac)}}. 
\]
As $m_\lambda^2$ denotes the new renormalized mass of the background field $\phi$, we need to impose the constraint that $m_\lambda^2 \geq 0$, otherwise the reference state for the system may not exist. This inequality always holds for sufficiently small $\lambda_1/g_1$, and it holds trivially in the case of vanishing expectation value of the Wick square. Besides, there is always the possibility of setting $\expvalom{\phi^2}^{\text{(bac)}}=0$ by means of the choice of the renormalization of the field $\phi^2$ (the constant $c$ in \eq \eqref{eq:point-splitting}). 
	
To simply further the analysis, we shall select the Minkowski vacuum state $\omega_0$ as reference state on $\cA$, whose corresponding two-point function is 
\[
    \omega_0(\phi(y)\phi(x)) = \expval{\phi(y)\phi(x)}{0} = \hbar \Delta_{+}(y-x), 
\]
where
\begin{equation}
\label{eq:Delta+}
	\Delta_+(y-x) \doteq \frac{1}{(2\pi)^3} \int_{\mathbb{R}^4} \delta(p^2 + m^2)\Theta(p_0) \e^{i p(x-y)} \d^4 p,
\end{equation}
$\Theta$ is the Heaviside step function, and $\delta$ the Dirac delta function. However, other choices of quantum states, such that $\expvalom{\phi^2}$ is regular in $\cM$, do not alter significantly our analysis.

\subsection{The linearized expectation value of the Wick square}

Using the Bogoliubov map given in \eq \eqref{eq:Bogoliubov}, the expectation value of the Wick square can be evaluated at every perturbation order. The linearized contribution defined in \eq \eqref{eq:phi2-bac-lin} in the adiabatic limit ($f =  1$) takes the following form
\begin{equation}
\label{eq:expval-phi-linear}
	\expvalom{\phi^2}^{\text{(lin)}} = -i \hbar \lambda \int_\cM \at \Delta_{F,\omega}^2(y,x) - \Delta_{+,\omega}^2(y,x) \ct  \psi_1(y) \d y,
\end{equation}
where $\hbar \Delta_{F,\omega}(y,x) = \expvalom{T \at  \phi(y)\phi(x) \ct}$ and $\hbar \Delta_{+,\omega}(y,x) = \expvalom{\phi(y)\phi(x)}$ are the Feynman propagator and the two-point function associated to $\omega$, respectively. 
Notice that the definition of the Feynman propagator $\Delta_{F,\omega} \doteq \Delta_{+,\omega}+i\Delta_A$, where $\Delta_A$ is the advanced propagator, employed here differs by a factor $i$ from others constructions. Furthermore, the squares of 
$\Delta_{F,\omega}$ and $\Delta_{+,\omega}$ correspond to the pointwise multiplication of the integral kernels of the two distributions. We shall discuss how to construct these products below. A diagrammatic representation of the propagators in the integrand at the right hand side of \eq \eqref{eq:expval-phi-linear} is given in Figure \ref{Fig:loop}. 

The various propagators of the theory, and in particular those appearing in \eq \eqref{eq:expval-phi-linear}, are in general not invariant under translation. However, they acquire translation invariance when they are referred to the Minkowski vacuum $\omega_0$. To keep this in mind, we shall denote the Feynman propagator and the two-point function referred to the Minkowski vacuum as $\hbar\Delta_F(y-x) = \expval{T \at\phi(y)\phi(x)\ct}{0}$ and $\hbar \Delta_{+}(y-x) = \expval{\phi(y)\phi(x)}{0}$, respectively.
Furthermore, we denote by  
$W \doteq \Delta_{+,\omega} - \Delta_{+}$ and we observe that $W$ is a smooth function whenever $\omega$ is an Hadamard state, namely the singular part of two-point function $\Delta_{+,\omega}$ is the same as the one in the Hadamard parametrix, or, equivalently, $\omega$ fulfils the microlocal spectrum condition \cite{Radzikowski1996micro,Brunetti2000micro}.
Hence, recalling that $\Delta_F=\Delta_++i\Delta_A$, we get
\[
	\Delta_{F,\omega}^2 - \Delta_{+,\omega}^2 = \Delta_{F}^2 - \Delta_{+}^2 + i 2\Delta_A W.
\]

\begin{figure}[!htb] \centering{\resizebox{.60\linewidth}{!}{
\includegraphics[width=.80\linewidth]{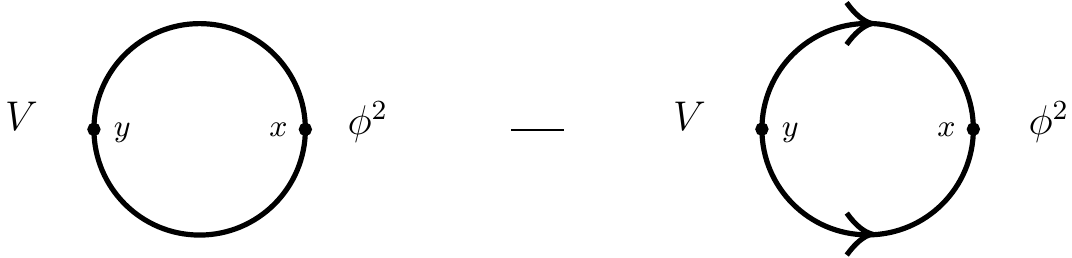}
}}
	\caption{\footnotesize Picture of the one-loop contribution $T \at V \phi^2 \ct - V\phi^2$ at the linear order in $\hbar$. The propagator $\Delta_F(y-x)$ is represented by a non-oriented line because it is symmetric in the exchange of $x \leftrightarrow y$, while the two-point function $\Delta_{+}(y-x)$ by an arrow from $y$ to $x$ \cite{Gere2016anly}.}
	\label{Fig:loop}
\end{figure} 
Since $W$ is smooth, we just need to construct $\Delta_+^2$ and $\Delta_F^2$ to give meaning to the right hand side of \eq \eqref{eq:expval-phi-linear}.
Contrary to $\Delta^2_{+}$ which is well defined everywhere, $\Delta_F^2$ is a well defined distributions only for test functions which are not supported on the origin. Therefore, a renormalization (extension) procedure is required to have a well defined $\Delta_F^2$ also in this case. The extension of $\Delta_F^2$ on test functions supported on the origin can be obtained keeping fixed the Steinmann scaling degree \cite{Steinmann1971pert,Brunetti2000micro}; however, this extension is not unique, and the remaining freedom amounts to an additional $c_0 \delta$, where $c_0$ is a real parameter. This freedom is compatible with the ambiguity present in the definition of $\phi^2$ \cite{Hollands2001wick}. To obtain explicit expressions of 
$\Delta_F^2$ and $\Delta_{+}^2$, we make use of the known K\"{a}llen-Lehmann spectral representations. In particular, using the convolution theorem and the definition of $\Delta_+(x)$ given in \eq \eqref{eq:Delta+}, we obtain that
\[
\Delta_{+}^2(x) = \int_{4m^2}^{\infty} \d M^2 \varrho(M^2) \Delta_+(x,M^2),
\]
where the spectral density $\varrho(M^2) = 
\frac{1}{(2\pi)^3}\int \frac{d^3{\mathbf{p}}} {\sqrt{\mathbf{p}^2+m^2}}
\delta(M-2\sqrt{\mathbf{p}^2+m^2})
=
\frac{1}{16\pi^2}\sqrt{1- \frac{4m^2}{M^2}}$,
and $\Delta_+(x,M^2)$ is the Minkowski vacuum two-point function given in \eq \eqref{eq:Delta+} with mass $M$. On the other hand, the very same representation for $\Delta_F^2(x)$ cannot be given because the integral over $M^2$ is divergent in that case. However, using the fact that for $x\neq0$, 
$(\Square+a)\Delta_F(x,M^2) = 
(M^2+a)\Delta_F(x,M^2)$, for $x\neq 0$ and $a>-4m^2$ we get
\[
    \Delta_{F}^2(x) = (\square+a)\int_{4m^2}^{\infty} \d M^2 \frac{\varrho(M^2)}{M^2+a} \Delta_F(x,M^2).
\]
This expression coincides with $\Delta_+^2(x)$ for $x \notin J^-(0)$ and with $\Delta_+^2(-x)$ for $x \notin J^+(0)$. Moreover, it is well defined also on functions whose support contains $0$, and thus it represents one of the possible extensions of $\Delta_F^2\in \mathcal{D}'(\cM\setminus \{0\})$ to $\mathcal{D}'(\cM)$. Also, the difference of two $\Delta_F^2$ constructed with $a$ and $a'\neq a$ is a non vanishing distribution supported in the origin with scaling degree $4$, and hence it must be proportional to the Dirac delta function. Finally, we can saturate the freedom in the construction of $\Delta_F^2$ with various choices of $a$. In light of this, the constant $a$ encodes the renormalization freedom present in the construction of $\Delta_F^2(x)$. Therefore, recalling that $\Delta_F(x)=\Delta_+(x) +i \Delta_A(x)$, the most general form of $\Delta_F^2(x) - \Delta_{+}^2(x)$ reads as
\begin{equation}
\label{eq:DeltaF2-Delta+2-dM2}
	\begin{aligned}
		-i \at \Delta_F^2(x) - \Delta_{+}^2(x) \ct =  (\square + a) \int_{4m^2}^{\infty} \d M^2 \frac{\varrho(M^2)}{M^2 + a} \Delta_A(x,M^2),
	\end{aligned}
\end{equation} 
where $\Delta_A(x,M^2)$ is the advanced propagator of the Klein-Gordon field of mass $M$. A detailed  derivation of \eq \eqref{eq:DeltaF2-Delta+2-dM2} for the massless case can be found in Appendix C of \cite{Duetsch2004causpert}. Hence, \eq \eqref{eq:expval-phi-linear} can be written at the linear order outside $x=0$ as 
\begin{equation}
\label{eq:expval-phi-linear-conv}
	\expvalom{\phi^2}^{\text{(lin)}} = {\hbar} \lambda \at \cK_a +  \cW \ct (\psi_1),
\end{equation}
where the operator $\cK_a$ maps compactly supported smooth function (or Schwartz function) to smooth functions. Moreover, its regularized integral kernel takes the form
\begin{equation}
\label{eq:kernel-K}
	\cK_a(x-y) = \int_{4m^2}^{\infty} \d M^2 \varrho(M^2) \frac{1}{M^2+a}(\square+a)\Delta_R(x-y,M^2) ,
\end{equation} 
where the d'Alembert operator is taken in the distributional sense, and $\Delta_R(\cdot,M^2)$ is the retarded propagator of the Klein Gordon equation with mass $M$. Its spatial Fourier kernel reads as
\begin{equation}
\label{eq:retarded-Fourier}
	\tilde{\Delta}_R(t,\vec{p},M^2) =  -\frac{\sin(\omega_0 t)}{\omega_0} \Theta(t),
\end{equation}
where $\omega_0 = \sqrt{|\vec{p}|^2 + M^2}$, and $\Theta$ is the Heaviside step function. The operator $\cW$ maps compactly supported smooth functions to smooth functions, and its integral kernel is the pointwise multiplication of the advanced propagator with the smooth part of the two-point function $W$, i.e., $\cW \doteq 2 \Delta_A W$. 

With the choice of the Minkowski vacuum as reference state, both $W$ and $\cW$ vanish, and thus the linearized expectation value of the Wick square given in \eq \eqref{eq:expval-phi-linear-conv} simplifies as
\begin{equation}
\label{eq:expval-phi-linear-vacuum}
	\expvalv{\phi^2}^{\text{(lin)}}= \hbar \lambda \; \cK_a(\psi_1).
\end{equation}
For later purposes, we need to control the evolution of $\expvalv{\phi^2}^{\text{(lin)}}$ in time under the influence of $\psi_1$: to this end, we study the kernel given in \eq \eqref{eq:kernel-K} in the Fourier domain by means of the following proposition. 
\begin{proposition}
	\label{prop:state}
	Let $\psi_1 \in \mathcal{S}(\cM)$ be a Schwartz function on $\mathcal{M}$, and let $\hat{\psi}_1$ be its Fourier transform.
	Then the Fourier transform of the linearized expectation value of the Wick square given in \eq \eqref{eq:expval-phi-linear-vacuum} can be written as
	\begin{equation}
	\label{eq:expval-phi-linear-vacuum-transf}
		\cF \ag \expvalv{\phi^2}^{\text{\normalfont(lin)}} \cg (p_0,\vec{p}) = \lim_{\epsilon\to 0^+}
		\frac{\lambda \hbar}{16\pi^2} F_a(-(p_0-i\epsilon)^2+|\vec{p}|^2) 
		\hat\psi_1(p_0,\vec{p}),
	\end{equation}
	given for strictly positive mass $m>0$ and for $a>-4m^2$. The function $F_a(z)$ admits the following integral representation:
	\begin{equation}
	\label{eq:F}
		\begin{aligned}
		F_a(z) &= \int_{4m^2}^{\infty}  \sqrt{1- \frac{4m^2}{M^2}} \at \frac{1}{M^2+a} - \frac{1}{M^2+z } \ct \d M^2,
		\end{aligned}
	\end{equation} 
	and it has the following properties:
	\begin{itemize}
	\item[a)] $F_a(z)$ is analytic for $z \in \mathbb{C} \setminus (-\infty,-4m^2]$ and continuous at $z=-4m^2$;
	\item[b)] the domain $F_a(z)$ has a branch cut on $z\in(-\infty,-4m^2)$ because there the imaginary part is discontinuous (the real part is continuous but not differentiable);
	\item[c)] $F_a(a)=0$;
	\item[d)] $F_a(s)$ is real for $s\in \aq -4m^2, \infty\ct$, it is strictly increasing for $s \in \aq -4m^2, \infty\ct$, and it diverges for large $|s|$; 
	\item[e)] The imaginary part of $F_a$ admits the following integral representation: 
	\[
		\text{\normalfont Im} (F_a(z))= \text{\normalfont Im}(z) \int_{4m^2}^{\infty}  \sqrt{1- \frac{4m^2}{M^2}} \at  \frac{1}{|M^2+z|^2 } \ct \d M^2,
	\] 
	it is strictly positive for $\text{\normalfont Im}(z)>0$, and strictly negative for $\text{\normalfont Im}(z)<0$.
	Furthermore, it vanishes for $z\in (-4m^2,\infty)$, and it is discontinuous on $z\in(-\infty,-4m^2)$ (the absolute value is finite).
	\end{itemize}
	Finally, for $z\not\in(-\infty,0)$ and $a>0$, $F_a(z)$ takes the form
	\begin{equation}
	\label{eq:FF}
		\begin{aligned}
			F_a(z) &= 2\sqrt{\frac{z+4m^2}{z}} \log\left(\frac{\sqrt{z+4m^2}+\sqrt{z}}{2m}\right)
		-2\sqrt{\frac{a+4m^2}{a}}\log\left(\frac{\sqrt{a+4m^2}+\sqrt{a}}{2m}\right).
		\end{aligned}
	\end{equation} 
\end{proposition} 
The qualitative behaviour of $\text{Re}(F_a(z))$ and of $|\text{Im}(F_a(z))|$ for $z\in ( -4m^2, \infty)$ is plotted in Figure \ref{Fig:f}. 
	\begin{figure}[!htb]
		\begin{minipage}{0.5\textwidth}
			\centering
			\includegraphics[width=.99\linewidth]{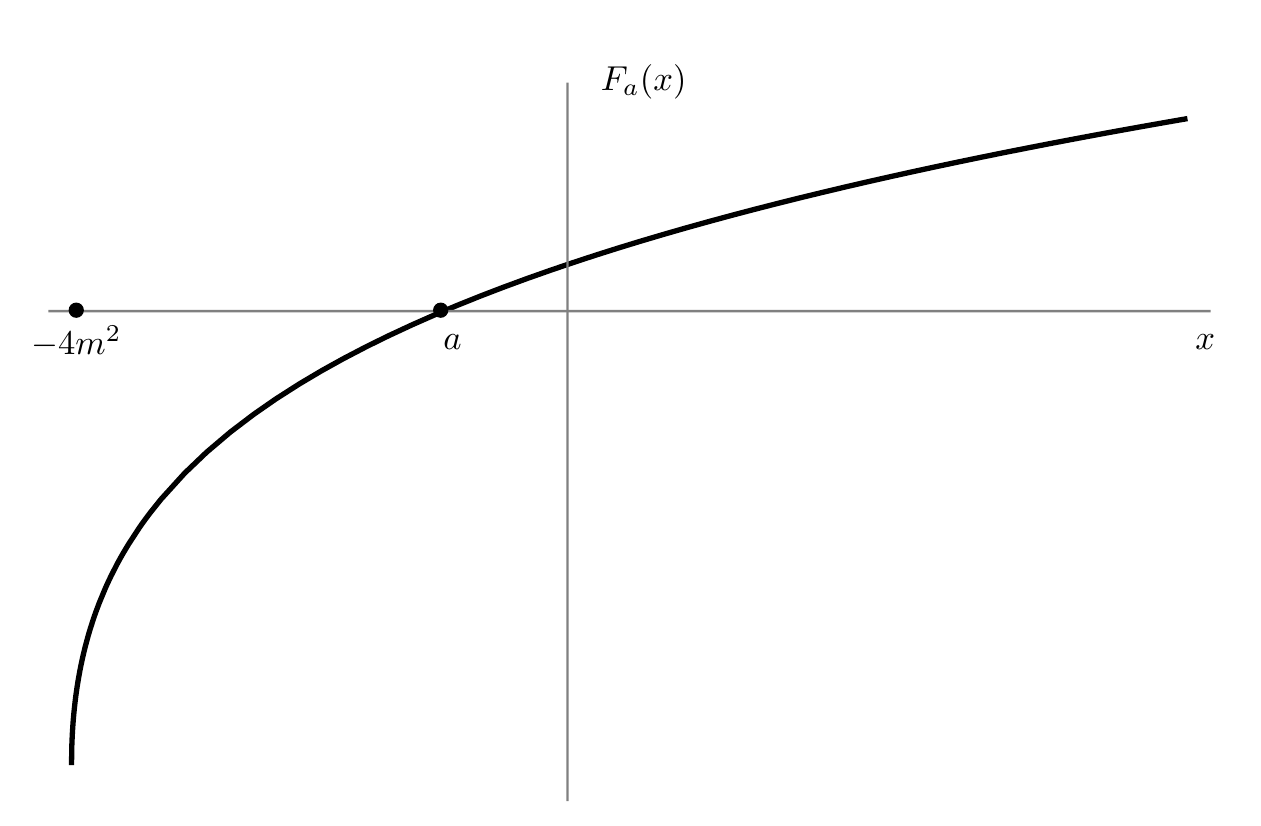}
		\end{minipage}\hfill
		\begin{minipage}{0.5\textwidth}
			\centering
			\includegraphics[width=.99\linewidth]{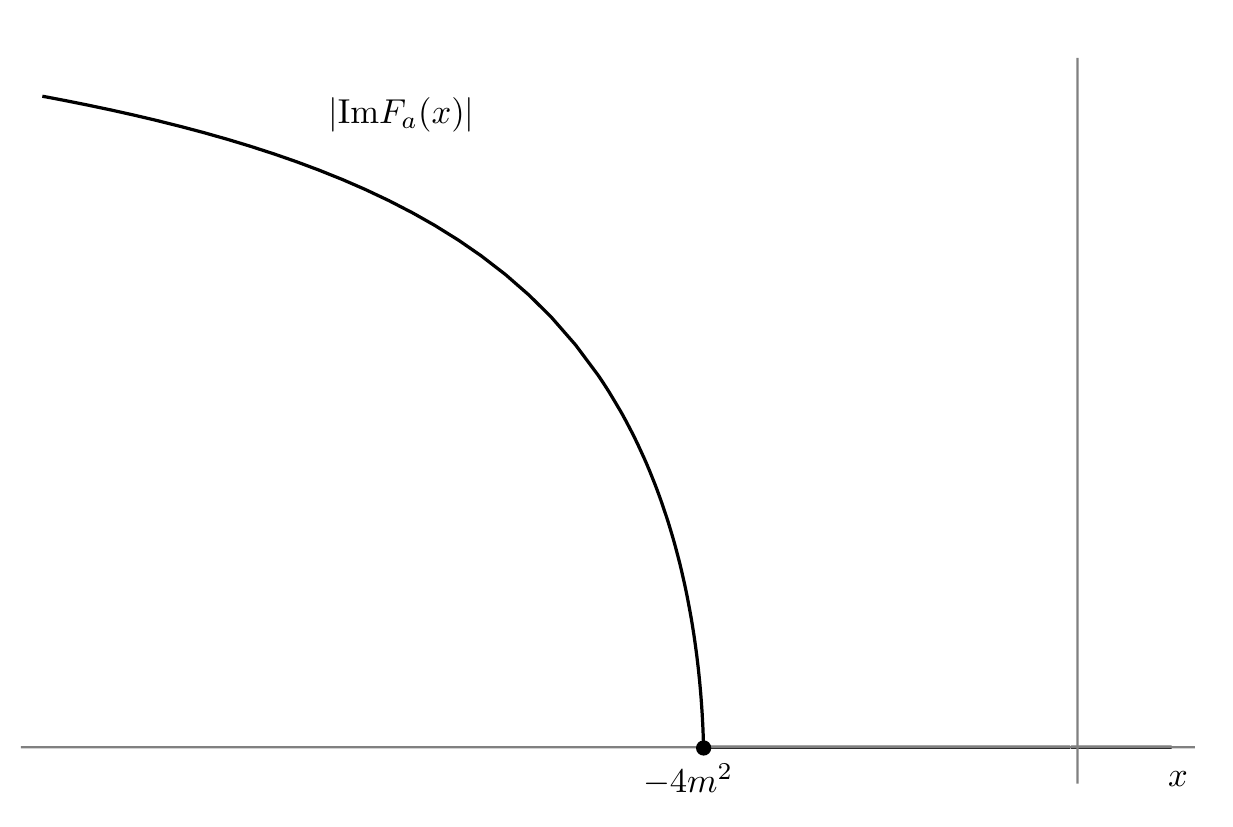}
		\end{minipage}
		\caption{\footnotesize The first graph contains the qualitative behaviours of $F_a(x)$ for $x \in (-4m^2,\infty)$ with $-4m^2<a<0$. The second graph is the qualitative behavior of $|\text{Im}F_a(x)|$.}\label{Fig:f}
	\end{figure}

\begin{proof}
	Using the definition of $\cK_a$ in \eq \eqref{eq:kernel-K}, the Fourier transform of the distribution $\Delta_R$ given in \eq \eqref{eq:retarded-Fourier} is
	\[
		\hat{\Delta}_R(p_0,\vec{p}) = \frac{1}{(p_0-i0^+)^2 - |\vec{p}|^2 -M^2}.
	\] 
	Hence taking the Fourier transform on both sides of \eq \eqref{eq:expval-phi-linear-vacuum} and
	applying the convolution theorem to $\cK_a(\psi_1)$ yields
	\begin{align*}
		&\cF \ag \expvalv{\phi^2}^{\text{(lin)}} \cg (p_0,\vec{p}) =  \\
		&\lim_{\epsilon\to 0^+}\lambda \hbar \int_{4m^2}^{\infty} \d M^2 \varrho_2(M^2) \frac{1}{M^2+a} \frac{1}{-(p_0-i\epsilon)^2 + |\vec{p}|^2 + M^2} \at -p_0^2 +|\vec{p}|^2  -a \ct \hat{\psi}(p_0,\vec{p}),
	\end{align*}
	where $\hat{\psi}_1(p_0,\vec{p}) \in \mathcal{S}(\cM)$. Then, the first part of the thesis follows after recalling the form of $\varrho_2(M^2) = \frac{1}{16\pi^2}\sqrt{1- \frac{4m^2}{M^2}}$ and the definition of $F_a(z)$.

The list of properties of $F_a(z)$ can be inferred directly from its integral representation. 
To check the validity of the representation \eqref{eq:FF} of $F_a$, consider 
\[
	A_a(z) \doteq \frac{1}{a-z}\left( 2\sqrt{\frac{z+4m^2}{z}} \log\left(\frac{\sqrt{z+4m^2}+\sqrt{z}}{2m}\right)
		-2\sqrt{\frac{a+4m^2}{a}}\log\left(\frac{\sqrt{a+4m^2}+\sqrt{a}}{2m}\right)  \right).
\]
From the expression of $F_a$ given in \eq \eqref{eq:FF}, $F_a(-p_0^2+|\vec{p}|^2)= (p_0^2-|\vec{p}|^2+a) A_a(-p_0^2+|\vec{p}|^2)$. We take the inverse Fourier transform in time of $A_a(-p_0^2+|\vec{p}|^2)$, that is, 
\[
	\tilde{\mathcal{A}}(t,\vec{p}) \doteq \lim_{\epsilon \to 0^+} \frac{1}{2\pi}\int_{-\infty}^\infty  A(-(p_0-i\epsilon)^2+\vec{p}^2) \e^{i p_0 t} \d p_0. 
\] 
The integrand in $\tilde{\mathcal{A}}$ has two cuts for $(p_0-i\epsilon)^2 > (|\vec{p}|^2+4m^2)$ located in the upper half complex plane and it is analytic outside the two cuts (it has no poles); furthermore, $|A(-(w-i\epsilon)^2+\vec{p}^2)|$ for $w\in\mathbb{C}$ vanishes in the limit $|w|\to\infty$. 
Hence, that inverse Fourier transform can be obtained by standard results of complex analysis, including Jordan's lemma and Cauchy residue theorem. In particular, to evaluate the integral over the real line, for $t<0$ we can close the contour in the lower half plane, and thus $\tilde{\mathcal{A}}=0$ because $A(-(w-i\epsilon)^2+\vec{p}^2)$ is analytic in the lower half plane.
On the other hand, if $t>0$, then the contour is closed in the upper half plane, and thus
only the two cuts matter in the evaluation of the integral over the real line which gives $\tilde{\mathcal{A}}$.
The contributions due to the two cuts for $t>0$ can be combined to give 
\[
	\tilde{\mathcal{A}}(t,\vec{p}) =  \lim_{\epsilon \to 0^+} \frac{1}{2\pi}\int_{\sqrt{|\vec{p}|^2+4m^2}}^\infty  
	\left(A(-w^2+\vec{p}^2+i\epsilon) -A(-w^2+\vec{p}^2-i\epsilon) \right) \left(\e^{i w t} - \e^{-i w t}\right) \d w. 
\]
Changing variable of integration to $M^2 = w^{2}-|\vec{p}|^2$ and computing the discontinuity of 
$\log\left(\frac{\sqrt{z+4m^2}+\sqrt{z}}{2m}\right)$ along its cut, 
we obtain for $t>0$ that
\[
	\tilde{\mathcal{A}}(t,\vec{p}) =  \frac{i2\pi}{2\pi}\int_{4m^2}^\infty \frac{1}{M^2+a} \left(\sqrt{1-\frac{4m^2}{M^2}}\right) \frac{\left(e^{i \omega_0 t} - e^{-i \omega_0 t}\right)}{2\omega_0}  \d M^2,
\]
where $\omega_0 = \sqrt{|\vec{p}|^2+M^2}$. So
\[
	{\mathcal{A}}(x) = \int_{4m^2}^\infty \frac{1}{M^2+a} \left(\sqrt{1-\frac{4m^2}{M^2}}\right) \Delta_R(x,M^2) \d M^2,
\]
and hence $(\square +a) \mathcal{A}$ is equal to $\mathcal{K}_a$ up to a constant factor given in \eq \eqref{eq:kernel-K}. Therefore,
 the expression of $F_a$ given in \eq \eqref{eq:F} follows, thus proving that it coincides with \eq \eqref{eq:FF}.
\end{proof}

\begin{remark}
\label{rem:massless}
In the limit of vanishing mass $m^2 = 0$ and for $-p_0^2+|\vec{p}|^2>0$, the function $F_a(-p_0^2+|\vec{p}|^2)$ given in \eq \eqref{eq:F} takes the form
	\begin{equation}
	\label{eq:F-massless}
		F_{a}(-p_0^2+|\vec{p}|^2)  = \log(\frac{-p_0^2+|\vec{p}|^2}{a}).
	\end{equation}
	This logarithmic behaviour is similar to the case studied in \cite{Horowitz1980wf} for the linearized semiclassical Einstein equations, in the weak-field limit of gravity and after considering massless quantum fields (see also \cite{Hartle1981ground}). The function $F_a$ is also similar to the Fourier transform of the Green function associated to the first order equation analyzed in \cite{Flanagan1996back} in the context of semiclassical back-reaction.
\end{remark}

\subsection{Linearized solutions}

The semiclassical equation \eqref{eq:SCE-psi1} governing the dynamics of the perturbations at the linear order is a linear equation in the perturbation field $\psi_1$. As we shall see below the properties of the solution space of that linear equation depends strictly on the parameters $a,\lambda,\lambda_i, g_i$, $i=1,2$. The non local state-dependent contribution $\expvalv{\phi^2}^{(\text{lin})}(x)$ defined in \eq \eqref{eq:expval-phi-linear-vacuum} was constructed in terms of the linear operator $\mathcal{K}_a$ introduced in  \eq \eqref{eq:expval-phi-linear-conv}, and thus it was expressed in terms of the function $F_a$ studied in Proposition \ref{prop:state}. Therefore, \eq \eqref{eq:SCE-psi1} evaluated in the Minkowski vacuum state reads
\begin{equation}
\label{eq:SCE-psi-linear}
	(g_2\square - g_1) \psi_1(x) = (\lambda_1 - \lambda_2 \square)\expvalv{\phi^2}^{(\text{\normalfont lin})}(x).
\end{equation}	
To highlight its mathematical structure, \eq \eqref{eq:SCE-psi-linear} can be rewritten in the following form: 
\begin{equation}
\label{eq:SCE-eq1-hom}
	\hbar \lambda P_\lambda \mathcal{K}_a(\psi_1)(x) + P_g \psi_1(x) = 0,
\end{equation}
where $P_\lambda \doteq \lambda_2 \square - \lambda_1$, and $P_g \doteq g_2 \square - g_1$.

\medskip

Because of the presence of a second d'Alembert operator in the expression of $\expvalv{\phi^2}^{(\text{lin})}(x)$ given in \eq \eqref{eq:expval-phi-linear-conv} through \eq \eqref{eq:kernel-K}, \eq \eqref{eq:SCE-psi-linear} contains fourth-order derivatives in $\psi_1$. 
It has thus a form similar to the semiclassical equations which usually appear in semiclassical theories of gravity, see, e.g., \cite{Kay1981see} and the next Section. Thus, it manifests the same conceptual issues already known in semiclassical gravity as higher-order theory of gravity. 
In particular, there are cases where similar equations admit the so-called runaway solutions, which make the classical background field unstable in the semiclassical approach \cite{Horowitz1978inst,Horowitz1980wf,Jordan1987stab,Parker1993red,Simon1991stab,Flanagan1996back}.

Runaway solutions usually consist of solutions of the linearized system around some background which grow exponentially in time. This class of linearized solutions become dominant over the background at large times. Thus, the full solution of the system acquires, in principle, a very different form from the chosen background, and at the same time it is expected to be very sensitive to the chosen initial conditions. Therefore, the background solution cannot be assumed to be stable. On the contrary, if all the linearized solutions decay sufficiently fast to zero for large times, then the perturbations become negligible with respect to the background solution, thus indicating the stability of the background.

In the next part we analyze \eq \eqref{eq:SCE-eq1-hom} equipped with a compactly supported smooth source term, namely
\begin{equation}
\label{eq:SCE-eq1-source}
	\hbar \lambda P_\lambda \mathcal{K}_a(\psi_1)(x) + P_g \psi_1(x) = f(x),
\end{equation}
where $\mathcal{K}_a$ is the linear operator introduced in \eq \eqref{eq:expval-phi-linear-conv}, and $f \in C_0^\infty(\cM)$ is a compactly supported source. 

The strategy is as follows. In a first step, we shall show that this equation manifests an hyperbolic nature, and we shall construct its retarded fundamental solutions as an operator $D_R: C^\infty_0(\mathcal{M}) \to C^\infty(\mathcal{M})$. 
Afterwards, thanks to the regularity properties of $D_R$, we shall prove that past compact solutions of the form $\psi_1 = D_R(f)$ decay as $1/t^{3/2}$ for large times $t$, 
hence getting the stability of the corresponding backgrounds against perturbation sourced by $f$.

In a second step, we shall study the smooth spatially compact solutions of the homogeneous equation \eqref{eq:SCE-eq1-hom} corresponding to \eq \eqref{eq:SCE-psi1}.
To determine uniquely a solution in the future and in the past of $t_0$, we shall equip \eq \eqref{eq:SCE-eq1-hom} with suitable initial conditions at $t_0 =0$, i.e., smooth compactly supported initial data of the form 
\[
	\psi_1^{(j)}\at 0,\vec{x} \ct = \varphi^j(\vec{x})
\]
for $j\in \{0,1\}$ or $j\in \{0,1,2,3\}$, with $\varphi^j\in C^\infty_0(\mathbb{R}^3)$.
We shall see that the number of initial conditions necessary to determine a spatially compact 
solution depends
on the choice of the parameters: in some cases, four initial conditions have to be imposed, while, in other cases, only two initial conditions are sufficient to obtain a solution; finally, there are also cases where no solution exists. Thus, we shall prove that there are wide ranges of values of $(a,g_1,g_2,\lambda,\lambda_1,\lambda_2)$ for which solutions of \eq \eqref{eq:SCE-eq1-hom} with compactly supported initial data decay faster than $1/t^{3/2}$ for large times. Therefore, the stability of the linearized back-reacted system is restored even in this case.

\medskip

Our analysis starts from showing that \eq \eqref{eq:SCE-psi-linear} written as \eq \eqref{eq:SCE-eq1-hom} manifests an hyperbolic nature.
\begin{proposition}
\label{theo:hyperbolic}
	Consider the equation \eqref{eq:SCE-eq1-source} sourced by $f\in {C}^\infty_0(\cM)$ in the form
	\[
	\hbar \lambda P_\lambda \mathcal{K}_a(\psi_1)(x) + P_g \psi_1(x) = f(x).
	\]	
	Set $\lambda_2\neq 0$ and $g_2\neq0$.
Let $\psi_1$ be a past compact solution of \eq \eqref{eq:SCE-eq1-source}, then
\begin{equation}
\label{eq:retardation}
	\psi_1 = 
 \frac{\lambda_2}{g_2}
 \left(
 \Delta_{R,\lambda} f - \hbar \lambda \mathcal{K}_a(\psi_1) - \Delta_{R,\lambda} \left(P_g-\frac{g_2}{\lambda_2}P_\lambda\right)\psi_1 \right)
\end{equation}
where $\Delta_{R,\lambda}$ is the retarded fundamental solution of $P_\lambda \Delta_{R,\lambda} = \mathbb{I}$.
Moreover, if we consider $\delta \psi_1, \delta f\in C^\infty_0(M)$ then, 
$\psi_1(x)+\delta \psi_1(x)$ is also a solution of 
\[
	\hbar \lambda P_\lambda \mathcal{K}_a(\psi_1+\delta \psi_1)(x) + P_g (\psi_1+\delta \psi_1)(x) = f(x)+\delta f(x).
	\]	
for every $x \not \in 
(J^{+}(\text{\normalfont supp} \delta f)
\cup J^{+}(\text{\normalfont supp} \delta \psi_1) )$.

\end{proposition}
\begin{proof}
As the solution $\psi_1$ is past compact by hypothesis, we can apply the retarded operator associated to  $P_\lambda$ on both sides of \eq \eqref{eq:SCE-eq1-source}, thus obtaining
\[
	\hbar \lambda \mathcal{K}_a(\psi_1) + \Delta_{R,\lambda} P_g\psi_1 = \Delta_{R,\lambda} f.
\]
Using the definition of fundamental solution $(\Delta_{R,\lambda} \circ P_\lambda) \psi_1 = \psi_1$, this equation can be written also as
\[
\hbar \lambda \mathcal{K}_a(\psi_1) + \Delta_{R,\lambda} \left(P_g-\frac{g_2}{\lambda_2}P_\lambda\right)\psi_1 + \frac{g_2}{\lambda_2}\psi_1 = \Delta_{R,\lambda} f, 
\]
and hence as
\[
	\frac{g_2}{\lambda_2}\psi_1 = \Delta_{R,\lambda} f - {\hbar} \lambda \mathcal{K}_a(\psi_1) - \Delta_{R,\lambda} \left(P_g-\frac{g_2}{\lambda_2} P_\lambda\right)\psi_1,
\]
thus yielding \eq \eqref{eq:retardation}. Notice that both $\Delta_{R,\lambda}$ and $\Delta_{R,\lambda} (\lambda_2 P_g-g_2 P_\lambda)$ have the retarded property, i.e., $\text{supp}(\Delta_{R,\lambda} f) \subset J^+ (\text{supp} f)$ and $\text{supp}(\Delta_{R,\lambda} (\lambda_2 P_g- g_2 P_\lambda) \subset  J^{+} (\text{supp} f)$. Similarly, $\mathcal{K}_a$ satisfies also the retarded propertied because it is an integral of retarded operators. 
The last observation descends from the fact that to compute $\mathcal{K}_a(\psi_1+\delta \psi_1)(x)$ only $\psi_1+\delta \psi_1$ in the past of $x$ matters.  Hence, if $x\not\in J^{+}(\text{supp}\delta f) \cup J^{+}(\text{supp}\delta \psi_1)$, $\psi_1(x)+\delta \psi_1(x) = \psi_1 (x)$, 
$f(x)+\delta f(x) = f (x)$,
and $\mathcal{K}_a(\psi_1+\delta \psi_1)(x)= \mathcal{K}_a(\psi_1)(x)$.
We conclude that, for $x\not\in J^{+}(\text{supp}\delta f) \cup J^{+}(\text{supp}\delta \psi_1)$, 
$\psi_1+\delta\psi_1$
is again a solution of the inhomogeneous equation with the source modified by $\delta f(x)$.

\end{proof}

The constraints on the parameters $\lambda_2$ and $g_2$ given by hypothesis in Proposition \ref{theo:hyperbolic} can be easily removed adapting the first part of the proof. For example, if $g_2=0$, there is no need to add $P_\lambda \Delta_{R,\lambda}$ in the right hand side of \eq \eqref{eq:retardation}. On the other hand, if $\lambda_2=0$, then the proof starts with getting an analog of \eq \eqref{eq:retardation} after applying the retarded operator $\Delta_{R,g}$ at the place of $\Delta_{R,\lambda}$ on both sides of \eq \eqref{eq:SCE-eq1-source}. Finally, if both $\lambda_2$ and $g_2$ vanish, then there is no need of preliminary applying any retarded operator to \eq \eqref{eq:SCE-eq1-source}.

Proposition   \ref{theo:hyperbolic}, 
and more precisely \eq \eqref{eq:retardation}, 
suggests that the form of a past compact solution $\psi_1$ of \eq  \eqref{eq:SCE-eq1-source} in $x$ cannot be influenced by any modification of $\psi_1$ or $f$ outside of $J^{-}(x)$. We shall see a posteriori that this indication is actually correct, because we shall prove in Theorem \ref{theo:retarded} that a retarded fundamental solution of \eq \eqref{eq:SCE-eq1-source} exists.   

\medskip

We proceed in the analysis of the form of the solution of the linearized semiclassical equation by studying the associated retarded fundamental solution. We use Fourier techniques to analyze the fundamental solution of \eq \eqref{eq:SCE-eq1-source}.
Hence,
\begin{equation}
\label{eq:linearized-source-fourier} 
	-\left((\lambda_1 + \lambda_2 (-(p_0)^2+|\vec{p}|^2))\frac{\lambda \hbar}{16\pi^2}  F_a(-(p_0-i0^+)^2+|\vec{p}|^2) +
	(g_2(-p_0^2+|\vec{p}|^2) + g_1)\right) \hat\psi_1(p_0,\vec{p})=\hat{f}(p_0,\vec{p}),
\end{equation} 
where $f\in C^{\infty}_0(\mathcal{M})$. \Eq \eqref{eq:SCE-eq1-source} can equivalently be written in a more compact form as
\[
	S(-(p_0-i0^+)^2+|\vec{p}|^2) \hat\psi_1(p_0,\vec{p})=\hat{f}(p_0,\vec{p}),
\]
where
\begin{equation}
\label{eq:S}
	S(z) \doteq - (\lambda_1+\lambda_2 z ) \frac{\lambda \hbar}{16\pi^2} F_a(z) - (g_1+g_2 z).
\end{equation}
In order to obtain the retarded fundamental solutions associated to the kernel $S(z)$, we need to study the set of points in the complex plane in which $S(z)$ vanishes: we denote this set by $\mathcal{S}$. Then, we shall prove that, if the parameters $(\lambda,g_i,\lambda_i)$ satisfy certain conditions, this set contains only real elements, and, furthermore, it includes only negative elements in some special cases. Among them, we shall impose as a constraint on the parameters the following inequality:
\begin{equation}
\label{eq:condition-param}
	g_2 \lambda_1-\lambda_2 g_1 \geq 0.
\end{equation}	
The characterization of elements in $\mathcal{S}$ is studied in the following proposition. 
\begin{proposition}
\label{prop:zeros}
	Let $\mathcal{S}\subset\mathbb{C}$ be set of zeros of $S(z)$ given in \eq \eqref{eq:S}. 
	Fix the parameters in such a way that at least one of the two $\lambda_i$, $i\in\{1,2\}$ is non vanishing and $g_2 \lambda_1-\lambda_2 g_1 \geq 0$, $\lambda>0$, $-4m^2<a$. Then we distinguish two cases:
	\begin{itemize}
		\item[a)] if $\lambda_2\neq0$,  and $-\lambda_1/\lambda_2 \leq -4m^2$, then $\mathcal{S}\subset(-4m^2,\infty)\cup\{-\lambda_1/\lambda_2\}\subset \mathbb{R}$;
		\item[b)] $\mathcal{S}\subset (-4m^2,\infty) \subset \mathbb{R}$ otherwise. 
	\end{itemize}
    In particular, if $\lambda_2\neq0$, $-\lambda_1/\lambda_2$ is in $ \mathcal{S}$ only if $g_2 \lambda_1/\lambda_2 = - g_1$. Furthermore, $\mathcal{S}$ contains one, two or no elements depending on the parameters $\lambda_i$, $g_i$, and $a$.
\end{proposition}
\begin{proof}
	Let us start with the equation $S(z)=0$ written in the form
	\begin{equation}
	\label{eq:CasoD}
		(\lambda_1+\lambda_2 z ) \frac{\lambda \hbar}{16\pi^2} F_a(z) = - (g_1+g_2 z).
	\end{equation}
To prove that the solution set is as stated in item a) and b), we proceed as follows. After multiplying both sides of the equation by $(\lambda_1+\lambda_2 \overline{z})$, taking the imaginary part yields the following equation:
\[
	|\lambda_1+\lambda_2 z |^2  \text{Im} \at \frac{\lambda \hbar}{16\pi^2} F_a(z) \ct = -  (g_2 \lambda_1-\lambda_2 g_1) \text{Im}(z),
\]
where the global sign of right-hand side depends only on $\text{Im}(z)$, because by hypothesis $g_2 \lambda_1-\lambda_2 g_1 \geq 0$. In particular, for strictly positive $\text{Im}(z)$ the right hand side is negative, while the left hand side is strictly positive thanks to item $e)$ of Proposition \ref{prop:state}; similarly, for strictly negative $\text{Im}(z)$ the right hand side is positive, while the left hand side is strictly negative.

Moreover, as $\text{Im}(F_a(x))$ is not $0$ also for $x < -4m^2$ (see Figure \ref{Fig:f}), we have that the only possible solutions of \eq \eqref{eq:CasoD} needs to be searched within  $(-4m^2,\infty) \cup \{-\lambda_1/\lambda_2\}$ when $\lambda_2\neq 0$, or in $(-4m^2,\infty)$ otherwise. Furthermore, 
 if both $\lambda_2\neq 0$ and $z=-\lambda_1/\lambda_2$, 
then the left hand side of \eq \eqref{eq:CasoD} vanishes, whereas the right hand side vanishes only when $g_2 \lambda_1/\lambda_2  =-g_1$. 

We are thus left with the analysis the following real equation:
\begin{equation}
\label{eq:CasoReale}
	(\lambda_1+\lambda_2 s ) \frac{\lambda \hbar}{16\pi^2} F_a(s) = - (g_1+g_2 s), \qquad  s\in (-4m^2,\infty) \cup \{-\lambda_1/\lambda_2\}. 
\end{equation}
The properties of the function $F_a(s)$ for $s\in  (-4m^2,\infty) \subset \mathbb{{R}}$ are listed in Proposition \ref{prop:state}, and the plot of this function is reported in Figure \ref{Fig:f}. 
Notice in particular that $F_a(s)$ is a concave function, because  the second derivative of the integrand in \eq \eqref{eq:F} is $-2\varrho(M^2) (M^2+s)^{-3}$, which is a strictly negative integrable function, and thus $F_a(s)''$ is strictly negative. Hence, for $\lambda_2=0$, $\lambda_1\neq 0$ by hypothesis, $(\lambda_1+\lambda_2 s)F_a(s)$ has a definite concavity. In this case, the maximum number of distinct solutions of \eq
\eqref{eq:CasoReale} is two. If $\lambda_2\neq 0$, we rewrite \eq \eqref{eq:CasoReale} as
\[
    \frac{\lambda \hbar}{16\pi^2} F_a(s) + \frac{g_2}{\lambda_2} = 
    \frac{1}{\lambda_2^2} \frac{\lambda_2 g_1-\lambda_1 g_2} {s + \frac{\lambda_1}{\lambda_2}} = I(s).
\]
Notice that $\frac{\lambda \hbar}{16\pi^2} F_a(s) + \frac{g_2}{\lambda_2}$ is monotonically increasing. Using the hypothesis that $\lambda_2g_1-\lambda_1g_2 \geq 0$, we observe that $I(s)$ is constant if 
$\lambda_2g_1-\lambda_1g_2 = 0$  or monotonically decreasing if $\lambda_2g_1-\lambda_1g_2 >0$; in this latter case, it has also a discontinuity (a vertical asymptote) in $s=-\lambda_1/\lambda_2$. Hence, also in this case there are at most two solutions, thus concluding the proof. 
\end{proof}

We observe that, under the hypothesis of Proposition \ref{prop:zeros}, $\mathcal{S}$, the space of zeros of \eq \eqref{eq:S}, coincides with the set of elements where \eq \eqref{eq:CasoReale} vanishes. Having established that there are at most two distinct solutions, owning the properties of the function $F_a(s)$ for $s\in  (-4m^2,\infty) \subset \mathbb{{R}}$ stated in Proposition \ref{prop:state}, and the plot of the qualitative behaviour of that  function reported in Figure \ref{Fig:f}, we may draw the following conclusion. There are cases where either one or two positive solutions of this equation exist, and there are cases where only one or two negative solutions exist. It is also possible to find cases where one positive and one negative solution exists. Finally, there are cases where no solutions exists at all.

Taking into account all the previous statements, we are now ready to write the explicit form of the retarded fundamental solution of \eq \eqref{eq:SCE-eq1-source}, and to show that past compact solutions decay to zero for sufficiently large times, as expected for a perturbation over a stable background.
\begin{theorem}
\label{theo:retarded}
	Consider the semiclassical equation with a source term $f\in C^\infty_0(\mathcal{M})$ given in \eq   \eqref{eq:SCE-eq1-source}
	in the form
	\[
		\hbar \lambda P_\lambda \mathcal{K}_a(\psi_1)(x) + P_g \psi_1(x) = f(x).
	\]
	Fix as non-vanishing constants at least one of the two $g_i$, and at least one of the $\lambda_i$, assume that the inequality $g_2 \lambda_1-\lambda_2 g_1 \geq 0$ holds, and set $-4m^2<a<0$.
	Suppose that the set $\mathcal{S}$ defined in Proposition \ref{prop:zeros} contains only real negative elements, then the Fourier transform of the retarded fundamental solution $D_R$ of \eq \eqref{eq:SCE-eq1-source} reads
	\[
		\hat{D}_R(p_0,\mathbf{p})= \frac{1}{S(-(p_0-i0^+)^2+|\mathbf{p}|^2)},
	\]
	where $S(z)$ was defined in \eq \eqref{eq:S}. Hence 
\begin{equation}
\label{eq:retarded-position}
	D_R(x) = -\sum_{s\in \mathcal{S}} \frac{1}{S'(s)} \Delta_R(x,-s) - \frac{\lambda \hbar}{16\pi^2} \int_{4m^2}^\infty  \sqrt{1-\frac{4m^2}{M^2}} \frac{ (\lambda_2 M^2-\lambda_1 )}{|S(-M^2)|^2} \Delta_R(x,M^2) \d M^2,
\end{equation}
where the elements of $\mathcal{S}$ are the zeros of $S(s)$, with $s\in (-4m^2,\infty) \cup \{-\lambda_1/\lambda_2\}$. The retarded fundamental solution $D_R$ is a linear operator which maps smooth compactly supported functions to smooth functions, and with this operator at disposal the solution $\psi_1$ of \eq \eqref{eq:SCE-eq1-source} with past compact support is
\begin{equation}
\label{eq:solution-ps1}
	\psi_1 = D_R(f).
\end{equation}
For $\lambda_2\neq 0$, $\psi_1(t,\vec{x})$  decays as $1/t^{3/2}$ for large $t$.
\end{theorem}

\begin{proof}
We analyze the equation \eqref{eq:SCE-eq1-source} in the Fourier domain. Using the results given in Proposition \ref{prop:state}, it takes the form of
\[
	S(-(p_0-i0^+)^2+|\vec{p}|^2) \hat\psi_1(p_0,\vec{p}) = \hat{f}(p_0,\vec{p}),
\]
where $S$ was defined in \eq \eqref{eq:S}. Thus, the Fourier transform of the retarded operator $D_R$ yields
\begin{equation}
	\hat{D}_R(p_0,\vec{p})= \frac{1}{S(-(p_0-i0^+)^2+|\vec{p}|^2)},
\end{equation}
and hence its Fourier inverse transform
\[
	\tilde{D}_R(t,\mathbf{p})=\lim_{\epsilon\to 0^+}\lim_{R\to\infty}\frac{1}{2\pi}\int_{-R}^R \frac{1}{S(-(z-i \epsilon)^2 +|\mathbf{p}|^2)} \e^{itz} \d z 
\]
can be evaluated by means of standard methods of complex analysis. 
In view of the properties of the function $F_a(z)$ given in \eq \eqref{eq:F}, the function $h(z)\doteq 1/S(-(z-i\epsilon)^2+|\vec{p}|^2)$ is defined in $\mathbb{C}\setminus \{\{z-i\epsilon\leq -\sqrt{|\mathbf{p}|^2+4m^2} \}\cup\{z-i\epsilon\geq  \sqrt{|\mathbf{p}|^2+4m^2} \} \}$. For $|z|>R$, it holds that $|h(z)|< l(|z|)$, where the function $l(r)$ vanishes in the limit of large positive $r$, because, in the worse case, $1/S(-z^2+|\mathbf{p}|^2)$ is dominated by $c/F_a(-z^2+|\mathbf{p}|^2)$ for large $|z|$, for some constant $c$, and $|F_a(-z^2+|\mathbf{p}|^2)|$ grows as $\log|z|$ for large $|z|$. Therefore, from Jordan's lemma we may close the contour $\gamma$ in the upper or lower plane, according to the sign of $t$, with a semicircle which does not contribute to the integral in the limit $R\to\infty$.

The function $1/S(-(z-i\epsilon)^2+|\mathbf{p}|^2)$ has two poles for each element $s \in \mathcal{S}$ (the set of  zeros of $S(z)$).
Since $s\in \mathcal{S}$  is negative by hypothesis, we have that the poles are located on the line $\text{Im}(z)=i \epsilon$, and correspond to the complex numbers
\[
	z= i\epsilon \pm \sqrt{|\mathbf{p}|-s}.
\]
Furthermore, the function $1/S(-(z-i\epsilon)^2+|\mathbf{p}|^2)$ has two branch cuts located at $z = x+i\epsilon $, where $x^2 \geq |\vec{p}|^2 + 4m^2$.
Thus, $1/S$ is analytic in the lower half plane, and hence, we obtain that $\tilde{D}_R=0$ for $t < 0$ by Jordan's lemma, because we close the contour in the lower half plane for $t<0$. 

On the other hand, if $t>0$, then we close the contour $\gamma$ in the upper half plane, and hence we need to take care of both the poles and the branch cuts. 
In this case, if we deform the previous contour $\gamma$ to a new $\tilde{\gamma}$ in such a way to avoid both the poles and the cuts, then the result of the contour integral over $\tilde{\gamma}$ vanishes. Therefore,  the only two non-vanishing contributions in $\tilde{D}_R$, denoted by $\tilde{O}$ and $\tilde{C}$, are due to the poles and the branch cuts, respectively.

The contribution due to the poles can be directly evaluated using the Cauchy residue theorem, which yields
\[
	\tilde{O}(t,\vec{p}) = -\frac{2\pi i }{2\pi}\sum_{s\in \mathcal{S}}  \frac{1}{S'(s)} \left(\frac{e^{ i w_s  t }- e^{ -i w_s t}}{2w_s} \right) = \sum_{s\in \mathcal{S}} \frac{1}{S'(s)}  \frac{\sin (w_st)}{w_s}, 
\]
where $w_s \doteq \sqrt{|\vec{p}|^2 - s}$, and hence, in view of \eq \eqref{eq:retarded-Fourier},
\begin{equation}
\label{eq:Ox}
	O(x) = -\sum_{s\in \mathcal{S}} \frac{1}{S'(s)} \Delta_R(x,-s).
\end{equation}
The contribution due to the cuts can be combined in the following form
\[
\tilde{C}(t,\vec{p}) =  \lim_{\epsilon\to 0^+ }\frac{1}{2\pi} \int_{\sqrt{|\vec{p}|^2+4m^2}}^\infty      
\left[\frac{1}{S(-p_0^2+|\vec{p}|^2+i\epsilon)}-\frac{1}{S(-p_0^2+|\vec{p}|^2-i\epsilon)}\right]
       \left( e^{i p_0 t } - e^{-i p_0 t }\right) \d p_0.
\]
Thus, recalling \eq \eqref{eq:retarded-Fourier} again, it can be written in the position domain as
\[
	C(x) =  \lim_{\epsilon\to 0^+ }\frac{-i}{2\pi} \int_{4m^2}^\infty \left[\frac{1}{S(-M^2+i\epsilon)}-\frac{1}{S(-M^2-i\epsilon)}\right] \Delta_R(x,M^2) \d M^2.
\]
The discontinuity in the two cuts is only due to the imaginary part of $F_a(z)$, hence
\begin{equation}
\label{eq:Cx}
	C(x) = {-} \frac{\lambda \hbar}{16\pi^2} \int_{4m^2}^\infty \sqrt{1-\frac{4m^2}{M^2}} \frac{ (\lambda_2 M^2-\lambda_1 )}{ | (\lambda_2 M^2 - \lambda_1 ) \frac{\lambda \hbar}{16\pi^2} F_a(-M^2) + (g_2 M^2-g_1)|^2} \Delta_R(x,M^2) \d M^2,
\end{equation}
thus getting \eq \eqref{eq:retarded-position} by combining \cref{eq:Ox,eq:Cx}. Furthermore, the integral over $M^2$ present in $C(x)$ in \eq \eqref{eq:Cx} can always be taken, even when both $\lambda_2$ and $g_2$ vanish, because $\hat{\Delta}_R$ and $1/|F_a(-M)|^2$ decay as $1/M^2$ and $1/|\log(M)|^2$ for large $M$, respectively.

\medskip

The decay of $\psi_1=D_R(f)$ for large $t$ descends straightforwardly from Lemma \ref{le:decay} applied to $D_R(f)$, which implies that $\Delta_R(f,-s)$ decays as $1/t^{3/2}$ for large $t$, with $s\leq 0$. Hence, the contribution $O(f)$ due the poles given in \eq \eqref{eq:Ox} has the desired time decay property. The same holds for the contribution $C(f)$ due to the cuts given in \eq \eqref{eq:Cx}, because for $\lambda_2\neq 0$ the function $M^2/|S(-M^2)|$ is integrable in $\d M^2$, and, furthermore, $f$ is smooth and compactly supported in time, so its time Fourier transform is a Schwartz function. 
\end{proof}

\begin{remark}
	\label{rem:cuts}
	From the form of the kernel of $D_R$ obtained in \eq \eqref{eq:retarded-position}, a generic past compact solution $\psi_1=D_R(f)$ defined in \eq \eqref{eq:solution-ps1} can be decomposed into two parts, so that
	\[
		\psi_1(x) = \psi_1^{O}(x) + \psi_1^C(x),
	\]
	where $\psi_1^{O} = -\sum_{s\in \mathcal{S}} \frac{1}{S'(s)} \Delta_R(f,-s)$ denotes the contribution due to the poles of $1/S$, while $\psi_1^C$ is the contribution due the cuts. We observe that, while there is a chance to determine $\psi_1^{O}$ by means of a finite number of initial conditions given at some time $t_1$ in the future of $\text{supp}f$, we expect that it is not possible to determine $\psi_1^C$ with a finite number of initial conditions, because the integration of $M^2$ is over uncountably many points.
	
	In spite of this fact, we notice that the homogeneous equation \eqref{eq:SCE-eq1-hom} may still, in some cases, give origin to a well-posed initial value problem to uniquely determine spatially compact solutions. Actually, the contribution due the cuts cannot enter the construction of the solutions of the homogeneous equation on the whole space.
	The reason is that the kernel of the multiplicative operator $T$, which acts on $\mathcal{S}(\mathbb{R}^4)$ and is defined as 
	\begin{align*}
		T(z) &\doteq \frac{S(z)}{\prod_{s\in\mathcal{S}}(z-s)},
	\end{align*}
	contains only $0$, with $z = -(p_0 - i\epsilon)^2 + |\vec{p}|^2$. Therefore, only the contributions due the poles can give origin to non trivial solutions of the homogeneous equation \eqref{eq:SCE-eq1-hom} written as $S(z) \hat{\psi_1} = 0$.
\end{remark}

The decay rate of the smooth past compact solutions $\psi_1$ proved in Theorem \ref{theo:retarded} is the same which was obtained by means of Strichartz estimates for real, massive quantum scalar fields in four-dimensional Minkowski spacetime \cite{Strichartz1977wave}. Actually, this behaviour is justified by the form of the fourth-order differential equation \eqref{eq:SCE-eq1-source}, which is composed of massive Klein-Gordon like operators on $(\cM,\eta)$. Thus, the retarded fundamental solution is still a combination of Klein-Gordon like fundamental solutions, and hence the past compact solutions given in \eq \eqref{eq:solution-ps1} inherit the same late-time behaviour estimated for real Klein-Gordon fields.

Eventually, we can now to discuss the solutions of the linearized semiclassical equation without source given in  \eq \eqref{eq:SCE-eq1-source}. 

\begin{theorem}
\label{theo:main}
	Consider the semiclassical equation \eqref{eq:SCE-psi-linear} written as 
	\[
	    (g_2\square - g_1) \psi_1(x) = (\lambda_1 - \lambda_2 \square)\expvalv{\phi^2}^{(\text{\normalfont lin})}(x).
	\]
	Fix as non-vanishing constants at least one of the two $g_i$, and at least one of the $\lambda_i$, assume that the inequality $g_2 \lambda_1-\lambda_2 g_1 \geq 0$ holds, and set $-4m^2<a<0$.
 
	Let $\mathcal{S}\subset (-4m^2,\infty) \cup \{-\lambda_1/\lambda_2\}\subset \mathbb{R}$ be the set of zeros of $S(z)$ given in \eq \eqref{eq:S}. As discussed in Proposition \ref{prop:zeros}, $\mathcal{S}$ contains one, two or no elements depending on the parameters $\lambda_i$, $g_i$ and $a$. If $\mathcal{S} = \emptyset$, then \eq \eqref{eq:SCE-psi-linear} admits no solutions. If $\mathcal{S} \neq \emptyset$, let $\psi_1(t,\vec{x})$ be a smooth solution of \eq \eqref{eq:SCE-psi-linear} with spatial compact support. Then its spatial Fourier transform is of the form
	\[
		\tilde{\psi}_1(t,\vec{p}) = \sum_{s\in\mathcal{S}} \left( C^s_+(\vec{p}) \e^{ + i t \sqrt{|\vec{p}|^2-s}}
		+C^s_-(\vec{p}) \e^{ - i t \sqrt{|\vec{p}|^2-s}}\right).
	\] 
	Moreover, if $\mathcal{S}$ contains only negative elements, then each solution $\psi_1$ of \eq \eqref{eq:SCE-psi-linear} is uniquely fixed by the initial values at $t=0$
	\[
		\psi_1^{(j)}\at 0,\vec{x} \ct = \varphi^j(\vec{x}), \qquad j\in \{0,\dots , 2 |\mathcal{S}|\},
	\]
	where $|\mathcal{S}|$ is the cardinality of $\mathcal{S}$, and $\varphi^j\in C^{\infty}_0(\mathbb{R}^3)$. Furthermore, in this case, $\psi_1(t,\vec{x})$ decays for large time at least as $1/t^{3/2}$. 
\end{theorem}
\begin{proof}
We analyze \eq \eqref{eq:SCE-psi-linear} in the Fourier domain. Using the results given in Proposition \ref{prop:state}, it takes the form:
\begin{equation}
\label{eq:main-theo} 
	\left((\lambda_1 + \lambda_2 (-p_0^2+|\vec{p}|^2)) 
	\frac{\lambda \hbar}{16\pi^2}  F_a(-(p_0-i0^+)^2+|\vec{p}|^2) + (g_2(-p_0^2+|\vec{p}|^2) + g_1) \right)\hat\psi_1(p_0,\vec{p})=0.
\end{equation} 
Let $\tilde{\psi}(t,\vec{p})$ be the spatial Fourier transform of a generic solution $\psi_1(t,\vec{x})$ of the form given in \eq \eqref{eq:solution-ps1}. This is a linear combination of $\e^{i p^j_0 t}$ with $z = -{p^j_0}^2+|\vec{p}|^2 \in \mathcal{S}$, where $\mathcal{S}$ is the set of points of the complex plane in which the function $S$ given in \eq \eqref{eq:S} vanishes, namely in which \eq \eqref{eq:CasoD} holds.

According to Proposition \ref{prop:zeros}, we have that $\mathcal{S}$ must be contained either in $(-4m^2,\infty)\subset \mathbb{C}$, or in $(-4m^2,\infty)\cup\{-\lambda_1/\lambda_2\}\subset \mathbb{C}$ for $\lambda_2\neq0$, $g_2 \lambda_1/\lambda_2 = - g_1$, and $\lambda_1/\lambda_2 \leq 4m^2$. 

\medskip

According to the number of negative solutions of \eq \eqref{eq:SCE-psi-linear}, the explicit form of $\psi_1(t,\vec{x})$ reads as follows.
If $\mathcal{S}$ contains only one negative solution $x=-\tilde{n}$, with $\tilde{n}\geq 0$, then any solution having smooth compactly supported initial data at $t=0$ is of the form
\[
	\psi_1(t,\vec{x}) = \int_{\mathbb{R}^3} \left(C_+(\vec{p}) \e^{ i w_{\tilde{n}} t } + C_-(\vec{p})\e^{ -i w_{\tilde{n}} t} \right) \e^{ i {\vec{p} \cdot\vec{x} }}  \d \vec{p}, 
\]
where $w_{\tilde{n}}(\vec{p})  = \sqrt{|\vec{p}|^2 +\tilde{n}}$, and $C_\pm$ are obtained from $({\hat{\varphi}}^0(\vec{p}),{\hat{\varphi}}^1(\vec{p}))$ by solving 
\[
	\begin{pmatrix} {\hat{\varphi}}^0\\ {\hat{\varphi}}^1 \end{pmatrix} =
	\begin{pmatrix}  1&1\\ i w_{\tilde{n}}&- i w_{\tilde{n}} \end{pmatrix} 
	\begin{pmatrix} C_+\\ C_- \end{pmatrix},
\]
which yields
\[
	\begin{pmatrix} C_+\\ C_- \end{pmatrix} = \frac{i}{2w_{\tilde{n}}}
	\begin{pmatrix} - i w_{\tilde{n}}&-1\\ -i w_{\tilde{n}}&1 \end{pmatrix} 
	\begin{pmatrix} {\hat{\varphi}}^0\\ {\hat{\varphi}}^1 \end{pmatrix}, 
\]
namely
\[
	C_+(\vec{p}) =  {\hat{\varphi}}^0(\vec{p})  - \frac{i}{2 w_{\tilde{n}}} {\hat{\varphi}}^1(\vec{p}), \qquad C_-(\vec{p}) =  {\hat{\varphi}}^0(\vec{p})  + \frac{i}{2 w_{\tilde{n}}} {\hat{\varphi}}^1(\vec{p}).
\]
Thus, thanks to Lemma \ref{le:decay}, we obtain the desired decay of $\psi(t,\vec{x})$ for large $t$.

If $\mathcal{S}$ contains only two distinct negative elements, four initial data are needed to fix the solution. Denoting with $s_1=-n_1$ and $s_2=-n_2$, $n_i\geq 0$, the two distinct elements of $\mathcal{S}$, 
 the linearized solution of the semiclassical equation \eqref{eq:SCE-psi-linear} is a combination of two 
solutions of the Klein Gordon equation with different square masses $n_i$. In this case, the solution with smooth compactly supported initial data $(\varphi^0(\vec{x}),{\varphi}^1(\vec{x}),\varphi^2(\vec{x}),\varphi^3(\vec{x}))$ at $t=0$ is of the form
\[
	\psi_1  (t,\vec{x}) = \int_{\mathbb{R}^3} \left(C^1_+(\vec{p}) \e^{ i w_{1} t } + C^1_-(\vec{p})\e^{ -i w_{1} t} \right)
	\e^{ i {\vec{p} \cdot\vec{x} }} \d \vec{p} + \int_{\mathbb{R}^3} \left(C^2_+(\vec{p}) \e^{ i w_{2} t } + C^2_-(\vec{p})\e^{-i w_{2} t} \right) \e^{ i {\vec{p} \cdot\vec{x} }} \d \vec{p},
\]
where $w_{i}(\vec{p}) = \sqrt{|\vec{p}|^2 +n_i}$, and $C^i_\pm(\vec{p})$ are obtained from $(\hat{\varphi}^0(\vec{p}),{\hat{\varphi}}^1(\vec{p}),\hat{\varphi}^2(\vec{p}),\hat{\varphi}^3(\vec{p}))$ by
solving 
\[
	\begin{pmatrix} \hat{\varphi}^0 \\ \hat{\varphi}^1 \\ \hat{\varphi}^2 \\ \hat{\varphi}^3 \end{pmatrix} =
	\begin{pmatrix}  1&1&1&1\\ i w_{1}&- i w_{1} & i w_{2}&- i w_{2} \\ - w_{1}^2&-  w_{1}^2 & - w_{2}^2&-  w_{2}^2 \\ -i w_{1}^3& + i w_{1}^3 & -i w_{2}^3& + i w_{2}^3   \end{pmatrix} 
	\begin{pmatrix} C^1_+\\ C^1_-\\ C^2_+\\ C^2_- \end{pmatrix}.
\]
The determinant of that matrix is equal to $-4(n_1-n_2)^2 w_1w_2$, and hence $C^i_\pm(\vec{p})$ can be written  as linear combinations of $(\hat{\varphi}^0(\vec{p}),{\hat{\varphi}}^1(\vec{p}),\hat{\varphi}^2(\vec{p}),\hat{\varphi}^3(\vec{p}))$ as
\[
	\begin{pmatrix} C^1_+\\ C^1_-\\ C^2_+\\ C^2_- \end{pmatrix} =  \frac{1}{2 (n_1-n_2)}
	\begin{pmatrix}  - w_2^2&i \frac{w_2^2}{w_1}&-1& \frac{i}{w_1}  \\ - w_2^2 &-i \frac{w_2^2}{w_1} & -1& -\frac{i}{w_1} \\ w_{1}^2&- i \frac{w_1^2}{w_2}& 1& -\frac{i}{w_2} \\ w_{1}^2&  i \frac{w_1^2}{w_2}& 1& \frac{i}{w_2}  \end{pmatrix} 
 	\begin{pmatrix} \hat{\varphi}^0 \\ \hat{\varphi}^1 \\ \hat{\varphi}^2 \\ \hat{\varphi}^3 \end{pmatrix}.
\]
Notice that these coefficients have either $w_1$ or $w_2$ in the denominator, in the worse case. Thus, the desired decay of $\psi(t,\vec{x})$ for large $t$ is obtained by applying Lemma \ref{le:decay} as before.

	\captionsetup[figure]{labelfont={bf},labelformat={default},name={Figure}}
	\begin{figure}[!htb] \centering{\includegraphics{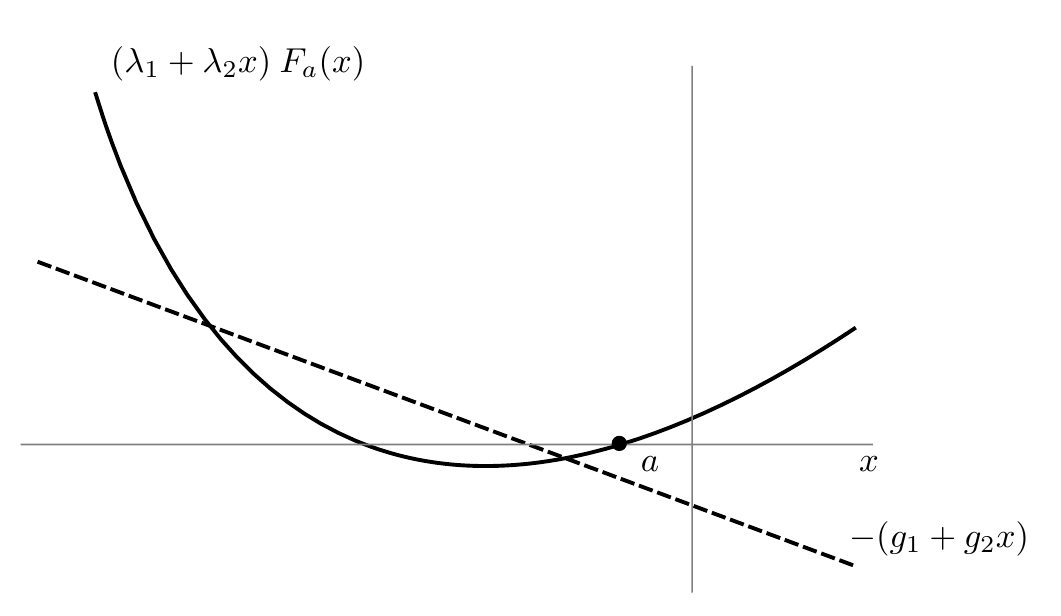}}
		\medskip
		\caption{\footnotesize Plots of the qualitative behaviours of $(\lambda_1+\lambda_2 x)(\frac{\lambda \hbar}{16\pi^2})F_a(x)$ and $(g_1+g_2 x)$ in $\aq -4m^2, \infty\ct$, where the constants are such that $g_i > 0$, $\lambda_i>0$, and $-4m^2<-\frac{\lambda_1}{\lambda_2}<-\frac{g_1}{g_2}<a<0.$}
		\label{Fig:fig_a}
	\end{figure}
	
	\begin{figure}[!htb] \centering{\includegraphics{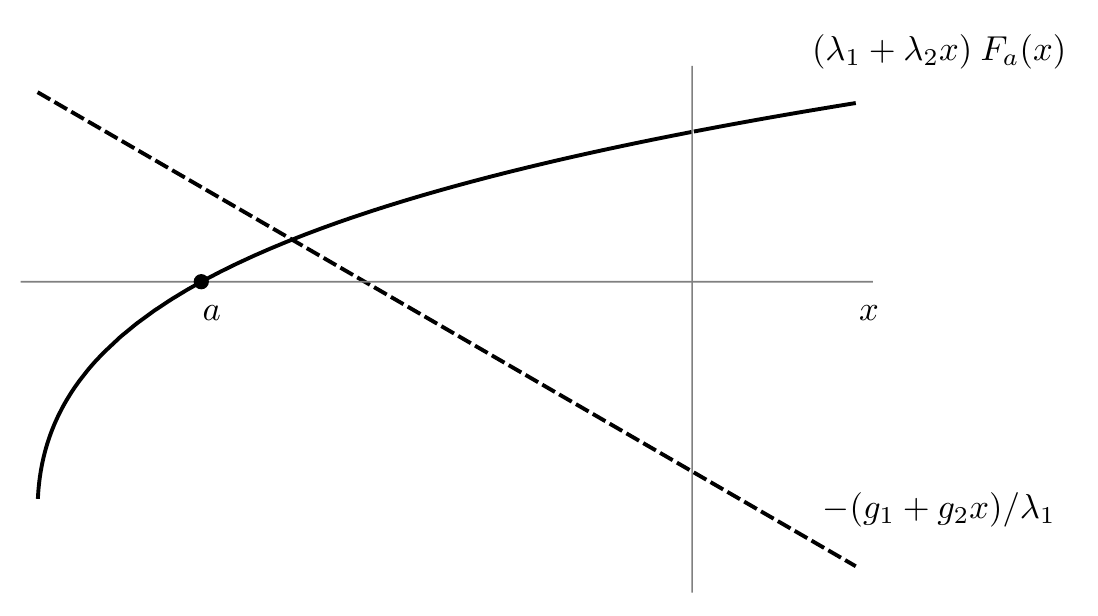}}
		\medskip
		\caption{\footnotesize Plots of the qualitative behaviours of $F_a(x)$ and $(g_1+g_2 x)/\lambda_1$ in $\aq -4m^2, \infty\ct$, where the constant are such that $g_i > 0$, $\lambda_2=0$, $\lambda_1>0$, $\frac{\lambda \hbar}{16\pi^2}F_a(-4m^2)<-(g_1-g_2 4m^2)/\lambda_1$, and $\frac{\lambda \hbar}{16\pi^2}F_a(0)>-g_1/\lambda_1$.}		
		\label{Fig:fig_b}
	\end{figure}	
\end{proof}

The previous theorem establishes that, if the space of solutions of \eq \eqref{eq:CasoReale} contains only negative elements, then all solutions of the linearized semiclassical equation \eqref{eq:SCE-psi-linear} with compactly supported initial values decay at large times. In the following corollary, we identify certain sufficient (but not necessary) conditions on the parameters $\lambda,\lambda_i,g_i,a$ which ensure such a behavior.
\begin{corollary}
\label{cor:cases}
Under the hypotheses of Theorem \ref{theo:main}, the space of solutions of \eq \eqref{eq:CasoReale} contains only negative elements if the following sufficient conditions on the parameters $\lambda,\lambda_i,g_i,a$ hold:

\begin{itemize}
	\item[a)] If $\lambda_2 = 0$, $\frac{g_2}{\lambda_1} \geq 0$, if $-\frac{g_1}{\lambda_1 } \leq \frac{\lambda \hbar}{16\pi^2} F_a(0)$, and $\frac{\lambda \hbar}{16\pi^2} F_a(-4m^2) \leq  -\frac{g_1}{\lambda_1}+4m^2 \frac{g_2}{\lambda_1}$, then $\mathcal{S}$ contains only one negative solutions.
	\item[b)] If $-\lambda_1/\lambda_2 < -g_1/g_2< a <0$, then $\mathcal{S}$ contains only negative solutions, and $|\mathcal{S}|$ is either $1$ or $2$.
	\item[c)] If $-4m^2<-\lambda_1/\lambda_2 < a <0$, $g_2>0$ and $0\leq -\frac{g_1}{\lambda_1 } \leq \frac{\lambda \hbar}{16\pi^2} F_a(0)$, then $\mathcal{S}$ contains only negative solutions, and $|\mathcal{S}|$ is either $1$ or $2$.
\end{itemize}
If these conditions hold, then the corresponding solutions of \eq \eqref{eq:SCE-psi-linear} with compact spatial support decay at least as $1/t^{3/2}$ at large times.
\end{corollary}

\begin{proof}
\begin{itemize}
	\item[a)] Case $\lambda_2 = 0$, $\frac{g_2}{\lambda_1} \geq 0$. \Eq \eqref{eq:CasoReale} takes the form 
	\[
		\frac{\lambda \hbar}{16\pi^2} F_a(x) = - \left(\frac{g_1}{\lambda_1 }+\frac{g_2}{\lambda_1 } x\right), \qquad  x\in (-4m^2,+\infty).
	\]
	The real part of that equation has now a positive solution if $-\frac{g_1}{\lambda_1 } > \frac{\lambda \hbar}{16\pi^2} F_a(0)$. From the plot displayed in Figure \ref{Fig:fig_b} we infer that a single negative solution appears if $-\frac{g_1}{\lambda_1 } \leq \frac{\lambda \hbar}{16\pi^2} F_a(0)$ and $ \frac{\lambda \hbar}{16\pi^2} F_a(-4m^2) \leq  -\frac{g_1}{\lambda_1}+4m^2 \frac{g_2}{\lambda_1}$, while no solutions exist otherwise.
	\item[b)] Case $a<0$, $-\lambda_1/\lambda_2 \leq -g_1/g_2\leq a <0$. From the plot displayed in Figure \ref{Fig:fig_a}, it is found that  all the possible solutions are negative, and in particular either one or two solutions exist.
	\item[c)] Case $-\lambda_1/\lambda_2 < a <0$  and $0\leq -\frac{g_1}{\lambda_1 } \leq \frac{\lambda \hbar}{16\pi^2} F_a(0)$. This case is similar to the one displayed in Figure \ref{Fig:fig_a}, but now the line $-g_2x + g_1$ intercepts the vertical line between $y=0$ and $y=\frac{\lambda \hbar}{16\pi^2} F_a(0)$. Hence, all the possible solutions are negative, and there is always at least one solution.
\end{itemize}

\medskip

Finally, the proof follows by applying the results of Theorem \ref{theo:main}.
\end{proof}

To summarize, we have thus proven that the desired decay as $1/t^{3/2}$ for large time $t$ of the solutions of the linearized Einstein equation \eqref{eq:SCE-psi-linear} with compact spatial support holds if and only if the zeros of $S$ defined in \eq \eqref{eq:S} are all contained in the negative real axis. On the contrary, if some zeros were located in the positive real axis, then unstable runaway solutions would destabilize the background configuration. Finally, if no zeros are present in $S$, then \eq \eqref{eq:SCE-psi-linear} does not admit solutions, but its counterpart with source given in \eq \eqref{eq:SCE-eq1-source} still has non vanishing past compact solutions which decay in time, due to the contribution given by the source through the branch cuts.

\begin{remark}
	According to the results presented in Proposition \ref{prop:zeros}, every solution of \eq \eqref{eq:CasoReale} is located in the positive real axis whenever the quantum field $\phi$ is massless, i.e., when $m=0$, even if the inequality \eqref{eq:condition-param} holds (recalling that $F_a(z)$ is reduced to \eq \eqref{eq:F-massless} in the massless case). For this reason, we may expect that any stability result cannot be achieved for massless fields, at least when homogeneous equations are taken into account, even if compactly supported initial data are selected.
	
	This observation is in accordance with the stability issue established in \cite{Horowitz1980wf}, where it was shown that exponentially growing runaway solutions appear when the back-reaction of a quantum Maxwell field interacting with a weak gravitational field is taken into account in the framework of semiclassical gravity.
\end{remark}
		
\subsection{Applications in the cosmological model}	

\label{sec:cosmo}

The analysis employed in the toy model presented in \eqs \eqref{eq:system-eq1} and \eqref{eq:system-eq2} can be used to guess the behaviour of the linearized solutions of the semiclassical Einstein equations in cosmological spacetimes, where matter is modelled by a massive quantum scalar field $\phi$. 

In the cosmological scenario we are considering here, the spacetime geometry is described by the metric $g_{\mu\nu}$ of the flat Friedmann-Lema\^itre-Robertson-Walker spacetime. In conformal coordinates $(\tau,x_1,x_2,x_3)$, the metric is conformally flat and reads
\begin{equation}
	\d s^2 = a(\tau)^2 \at -\d \tau^2 + \d \vec{x}^2 \ct,
\end{equation}
where the scale factor $a(\tau)$ of the universe represents the unique degree of freedom of the spacetime.
The dynamics of this universe is governed by the back-reaction of a linear quantum scalar field $\phi$, whose equation of motion is
\begin{equation}
\label{eq:scalar-field}
	\square \phi -m^2 \phi -\xi R \phi = 0,
\end{equation}
where $\xi$ denotes the coupling constant to the scalar curvature. Since there is an unique degree of freedom in this class of spacetimes, the dynamics of $a$ is determined, up to a constraint, by the trace of the semiclassical Einstein equations \cite{Meda2021exi}
\begin{equation}
\label{eq:SCE-cosmo}
	- R + 4\Lambda = 8 \pi G \expvalom{T}.
\end{equation}
The expectation value of the trace of the quantum stress-energy tensor associated to $\phi$ in a quantum state $\omega$ is 
\begin{equation}
\label{eq:traceT}
	\expvalom{T} = \left(3\left(\xi - \frac{1}{6}\right) \square - m^{2}\right)\expvalom{\phi^2} + \frac{1}{4 \pi^{2}}\left[v_{1}\right] + \alpha_1 m^4 - \alpha_2 m^2 R + \alpha_3 \square R,
\end{equation}
where the renormalization constants $\alpha_i$ are not fixed by the model, but describe the regularisation freedom present in the definition of the stress-energy tensor as normal-ordered Wick observable \cite{Hollands2001wick,Hollands2005stress,Hack2016cosm}. 
The constants $\alpha_1$ and $\alpha_2$ correspond to renormalizations of the cosmological constant and the Newton constant, respectively, and thus they can be reabsorbed in a redefinition of $\Lambda$ and $G$; on the contrary, $\alpha_3$ is of pure quantum nature and cannot be reabsorbed in any corresponding classical parameter of the theory.

Moreover, the coefficient $[v_1]$ appearing in \eq \eqref{eq:traceT} corresponds to the so-called quantum trace anomaly of the model, and reads 
\begin{equation}
\label{eq:trace-anomaly}
	[v_1] = \frac{1}{2880\pi^2 }\left(C^{abcd}C_{abcd} + R^{ab}R_{ab} -\frac{1}{3} R^2\right),
\end{equation}
up to contributions which can be reabsorbed by a redefinition of $\alpha_i$. Here, $C_{abcd}$ is the Weyl tensor, $R_{ab}$ the Ricci tensor and $R$ the Ricci scalar \cite{Wald1978trace,Moretti2003stress,Hollands2005stress}. 

We are interested in studying the linearized perturbation of this cosmological system around a spacetime which is a solution of the semiclassical Einstein equations written in the form of \eq \eqref{eq:SCE-cosmo}. As a first step of this analysis, and for the sake of simplicity, we consider as background solution a spacetime with vanishing curvature, namely the Minkowski spacetime. Under this assumption, we can obtain a formal correspondence between the linearization of \eq \eqref{eq:SCE-cosmo} and the linearized semiclassical equation \eqref{eq:SCE-psi-linear}, viewing $R$ as the perturbative external field $\psi_1$ over a vanishing background $\psi_0=0$. In view of this correspondence, the cosmological constant $\Lambda$ is a zeroth-order contribution which can be assumed to vanish, whereas the trace anomaly given in \eq \eqref{eq:trace-anomaly} is at least quadratic in the components of the Riemann curvature tensor $R_{abcd}$, and thus it is negligible at linear order in $R$. 

Taking into account all of this and \eq \eqref{eq:traceT}, \eq \eqref{eq:SCE-cosmo} takes the form of the linearized semiclassical equation \eqref{eq:SCE-psi-linear} through the following correspondence between the cosmological parameters and the set of constants $(g_1,g_2,\lambda_1,\lambda_2)$:
\begin{equation}
\label{eq:cosmo-param}
	g_1 = -\frac{1}{8\pi G}, \qquad g_2 = \alpha_3, \qquad \lambda=\xi, \qquad \lambda_1 = m^2, \qquad \lambda_2 = 3\left(\xi -\frac{1}{6}\right).
\end{equation}

In the cosmological framework, $g_1$ turns to be a fixed negative parameter (in Planck's units, $\hbar = 1$ and $(8\pi G)^{-1} = m_P^2/8\pi$, where $m_P$ is the Planck mass), while, on the other hand, $\lambda$ can be fixed to be strictly positive and different from $1/6$ by assuming non-minimally and non-conformally coupled fields, i.e, $\xi \neq 0, 1/6$; hence, $\lambda_2 \neq 0$. On the contrary, both $g_2$ and $\lambda_2$ are free parameters of the semiclassical theory, whose signs can be chosen such that the inequality \eqref{eq:condition-param} holds, i.e.,
\begin{equation}
\label{eq:cosmo-condition}
	\alpha_3 \frac{m^2}{m_P^2} \geq - \frac{3}{8\pi} \left(\xi -\frac{1}{6}\right), \qquad \alpha_3 \in \mathbb{R}.
\end{equation}
Under these assumptions, there are choices of the parameters $(m^2,\xi,a)$ for which the cosmological version of \eq \eqref{eq:CasoReale} admits only negative solutions. For example, by choosing $\xi > 1/6$, $\alpha_3 > 0$, $a > -4m^2$, and sufficiently large $m^2$ we may apply Corollary \ref{cor:cases}, which ensures that only negative solutions exist.
Namely, in these cases solutions of the cosmological linearized semiclassical Einstein equations written as in \eq \eqref{eq:SCE-cosmo} with spatial compact support decay to zero for large times, thus showing the stability of the chosen background. 

On the other hand, one may expect that too large values of $m^2$, even beyond the Planck scale $m_P$, would be physically unacceptable for quantum fields describing elementary particles. In this viewpoint, the result is similar to the one obtained in \cite{Randjbar-Daemi981stab} for massive quantum scalar fields in flat spacetime. Firstly, the conditions stated here are only sufficient, so other cases which provide stable solutions cannot be excluded a priori, for different choices of the parameters $(m^2,\xi,\gamma,a)$. Secondly, and most importantly, it is expected that the linearized perturbations in a more realistic cosmological model should be sourced by $f\in {C}^\infty_0(\cM)$ localized somewhere in the past. This source may have a quantum origin, for instance, related to some anisotropic or stochastic fluctuations at microscopic levels. This is the case which occurs for example in Stochastic Gravity, where a noise kernel bi-tensor modelling the stress-energy tensor fluctuations is added to the semiclassical Einstein equations, obtaining in this way the so-called Einstein-Langevin equations (see \cite{Hu2020stoc} and reference therein). In this picture, the stochastic source in the past drives the fluctuations of the gravitational field, and thus gives origin to the external perturbation which enters the cosmological linearized semiclassical Einstein equations as external source.

In this model with external sources, the cosmological counterpart of the linearized semiclassical equation \eqref{eq:SCE-eq1-source} should be taken into account, with parameters $\lambda_1,\lambda_2,g_1,g_2$ fixed as in \eq \eqref{eq:cosmo-param} and satisfying the inequality \eqref{eq:cosmo-condition}. 
Under these assumptions, and based on the results shown in Theorem \ref{theo:retarded} and Remark \ref{rem:cuts}, the linearized curvature solution $R$ depends on both the contributions due to the poles and the branch cuts of $S(z)$. However, the contribution arising from poles are not present for several, apparently more physically acceptable values of the parameters $(m^2,\xi,a)$: for example, for sufficiently large ratio $m_P^2/m^2$ and
\[
	0 < \xi < 1/6, \qquad \alpha_3 > 0,  \qquad a > -4m^2,
\]
the condition \eqref{eq:cosmo-condition} holds. 
Furthermore, with this choice of parameters, the past compact linearized solution induced by a smooth compactly supported source $f$ has no poles contribution, and hence it decays to zero for large times according to the results stated in Theorem \ref{theo:retarded}.

\section{Conclusions}
\label{sec:conc}

In this paper, we have analyzed the stability problem of semiclassical theories in flat spacetime, using a toy model consisting of a quantum scalar field coupled to a second entirely classical scalar field. This toy model mimics other semiclassical theories of gravity described by the semiclassical Einstein equations, where a quantum matter field propagates over a classical curved background. 
It is known that higher order derivatives appearing in the semiclassical equations can destabilize the system, giving rise exponentially growing linearized solutions (see the references given in the \hyperref[sec:intro]{Introduction}). 
The main result stated in this paper consists of proving that, if the quantum field driving the back-reaction is massive, then the stability of background solutions can be restored at the linear order in the interaction, for spatially compact perturbations and for large values of the coupling constants, after assuming some sufficient (but not necessary) conditions on the parameters of the theory. On the other hand, removing the assumption of massive quantum fields seems to give rise to runaways solutions which may alter stability, in accordance with other results present in the literature about semiclassical theories.

Namely, it is shown that unique solutions of this semiclassical initial value problem tend to disperse in time, namely at fixed position in space they decay polynomially in time. 

As this toy model mimics the dynamics of the linearized semiclassical Einstein equations, our analysis indicates a possible mechanism to get stability in several linearized semiclassical theories of gravity, even for different conditions than the ones stated in this paper. For example, in the case of a semiclassical theory in cosmological spacetimes, which has been already investigated by the authors in \cite{Pinamonti2011init,Pinamonti2015glob,Meda2021exi}. 

\bigskip

\subsection*{Acknowledgements}
We thank the anonymous referees for helpful comments on an earlier version of this paper. The work of P.M. was supported by a PhD scholarship of the University of Genoa. We are grateful for the support of the National Group of Mathematical Physics (GNFM-INdAM).

\bigskip

\appendix

\section{Large time decays of certain functions}
\label{sec:appendix}

This appendix contains a Lemma with the proof of the decay at large times of certain functions. The main idea of this proof is already known in the literature, see, e.g., \cite{Bros2002thermal}, however 
since this Lemma is a key result for the stability discussed in the main text, we report
its proof here for completeness.
\begin{lemma}\label{le:decay}
Let  $\tilde{f} \in \mathcal{S}(\mathbb{R}^3)$ be a Schwartz function, and consider its spherical average
$f(|\vec{p}|) = \frac{1}{4\pi}\int \tilde{f}(\vec{p})  \d \Omega$, where $\Omega$ denotes the standard measure on the two-dimensional surface of the unit sphere.  Consider  
\[
	W(t) = \int_{0}^\infty {f(p)} \frac{\e^{iw_{m^2}t}}{w_{m^2}}   p^2 \d p, \qquad w_{m^2} = \sqrt{p^2+m^2},
\]
then $W(t)$ decays for large times as $1/t^{3/2}$ if $m>0$, and as $1/t^2$ if $m=0$.
\end{lemma}
\begin{proof}
Let us start discussing the massless case $m=0$, In this case $w_0=p$ then
\[
	W(t) = \int_{0}^\infty {f(p)} \e^{i p t}   p \d p = \frac{1}{(it)^2}\int_{0}^\infty {f(p)}  p  \partial_p^2  \e^{i p t}   \d p.
\]
After integrating by parts twice, and using the rapid decay of $f \in \mathcal{S}(\mathbb{R}^3)$, 
\begin{equation}
\label{eq:W(t)}
	W(t) = \frac{f(0) }{(it)^2} + \frac{f(0) }{(it)^2} \int_{0}^\infty  \partial_p^2 ( pf(p))  \e^{i p t}   \d p,
\end{equation}
where $\partial_p^2 ( pf(p)) \in L^1(\mathbb{R}^+)$ because $f(p)$, together with its derivatives, is of rapid decrease. Hence, by Riemann-Lebesgue lemma 
\[
	\lim_{t\to\infty }\int_{0}^\infty  \partial_p^2 ( pf(p))  \e^{i p t}   \d p =0,
\]
which implies that the second contribution in \eq \eqref{eq:W(t)} vanishes more rapidly than $1/t^2$ for large times.

\medskip
We pass now to analyze the case $m>0$. After changing variable of integration $p=\sqrt{(y/t+2m)(y/t)}$, 
\[
	W(t) = \frac{\e^{imt}}{t^{3/2}} \int_{0}^\infty {f\left(\sqrt{\frac{y}{t}\left(\frac{y}{t}+2m\right)}\right)} {\e^{iy }} \sqrt{\frac{y}{t}+2m}  \sqrt{y} \d y.
\]
To evaluate this integral, we insert an $\epsilon$ regulator, and we divide the integral in two parts, 
\[
	W(t) = \frac{e^{imt}}{t^{3/2}}f(0) \sqrt{2m}  \lim_{\epsilon\to 0^+} \int_{0}^\infty \e^{iy -\epsilon y }  \sqrt{y} \d y +  \frac{\e^{imt}}{t^{3/2}} \lim_{\epsilon\to 0^+} \int_{0}^\infty \left(g(\frac{y}{t})-g(0)\right) \e^{iy - \epsilon y} \sqrt{y} \d y,
\]
where $g({y})= \sqrt{{y}+2m}f(\sqrt{{y}({y}+2m)})$. Notice that the integral in the first contribution tends to a constant in the limit $\epsilon \to 0$, because
\[
	\lim_{\epsilon \to 0^+} \int_{0}^\infty \e^{iy -\epsilon y }  \sqrt{y} \d y = \frac{\sqrt{\pi}}{2 (-i)^{3/2}}. 
\]
The second contribution decays faster then $1/t^{3/2}$. Actually, 
consider
\[
	A \doteq 
	\lim_{\epsilon\to 0^+} \int_{0}^\infty \left(g(\frac{y}{t})-g(0) \right) \e^{iy} \e^{-\epsilon y}\sqrt{y} \d y
\]
which is equal to 
\[
	A = \lim_{\epsilon\to 0^+} \frac{1}{(i-\epsilon)^2}\int_{0}^\infty \sqrt{y} \left(g(\frac{y}{t})-g(0)  \right) \partial_{y}^2 \left( e^{iy -\epsilon y}-1\right) \d y.
\]
After integrating by parts twice, and in view of the decay properties of $g$ and its derivatives for large arguments, 
\[
	A = \lim_{\epsilon\to 0^+} \frac{1}{(i-\epsilon)^2}\int_{0}^\infty \partial_{y}^2 \sqrt{y} \left(g(\frac{y}{t})-g(0)  \right) \left(\e^{iy -\epsilon y}-1\right) \d y.
\]
Hence, for $0<\delta<1/2$ we obtain that
\[
	A = \lim_{\epsilon\to 0^+} \frac{1}{(i-\epsilon)^2} \frac{1}{t^{\delta}}\int_{0}^\infty h_{\delta}\left(\frac{y}{t}\right) \frac{\left(\e^{iy -\epsilon y}-1\right)}{y^{3/2-\delta}} \d y,
\]
where
\[
	h_\delta(x) \doteq \left(x\right)^{3/2-\delta}\partial_{x}^2 \left(\sqrt{x} \left(g(x)-g(0)\right) \right)
\]
is a bounded function, because ($g(x)-g(0)$) vanishes as $x$ for $x$ near $0$, and $g$ decays faster to 0 for large $x$. Therefore, 
\[
	|A| \leq \lim_{\epsilon \to 0^+} \frac{C}{t^{\delta}} \int_{0}^\infty \left|\frac{\left( \e^{iy -\epsilon y}-1\right)}{y^{3/2-\delta}}\right| \d y
\]
for a suitable constant $C$. Since $C$ is uniform in $\epsilon$, and $0<\delta<1/2$, we can apply dominated convergence theorem to take
the limit as $\epsilon\to 0$ before computing the integral, and hence
\[
	|A| \leq  \frac{\tilde{C}}{t^{\delta}},
\]
which concludes the proof.

\end{proof}

\printbibliography

\end{document}